\ifdefined\fullversion
\else
\def\fullversion{0}    
\fi

\ifdefined\cameraversion
\else
\def\cameraversion{0}    
\fi

\def\allow{1}      

\documentclass[envcountsame,runningheads,notitlepage]{llncs}
\ifnum\fullversion=1
\fi

\usepackage{bm}
\usepackage{chngpage}
\usepackage{algpseudocode}
\usepackage{amsmath}
\usepackage{tikz}
\usepackage{pgfplots}
\pgfplotsset{compat=1.18}
\usetikzlibrary{arrows.meta}

\usepackage{float}
\usepackage{newfloat}
\DeclareFloatingEnvironment[
    fileext=los,
    listname=List of Schemes,
    name=Scheme,
    placement=H,
    within=section,
]{scheme}

\ifdefined\XeTeXversion
  \usepackage{pifont}
  \newcommand{\contactenvelope}{\ding{41}}
\else
  \usepackage{fontawesome5}
  \newcommand{\contactenvelope}{\faIcon[regular]{envelope}}
\fi

\newcommand{\emailicon}[1]{%
  \href{mailto:#1}{%
    \textsuperscript{\scriptsize\contactenvelope}%
  }%
}

\usepackage{adjustbox}
\usepackage{blindtext}
\usepackage{placeins}

\ifnum\cameraversion=0
\fi

\usepackage[utf8]{inputenc}
\usepackage[T1]{fontenc}
\usepackage{hyperref}
\usepackage{verbatim}
\usepackage{tikz}
\usetikzlibrary{positioning,calc,decorations.pathreplacing,backgrounds}
\usepackage{pgfplots}
\usepackage{pgfplotstable}
\pgfplotsset{compat=1.18}
\usepackage{subcaption}
\usepackage{forest}
\usepackage{xspace}
\usepackage{amsmath} 
\usepackage{amssymb}
\usepackage{mathtools}
\usepackage{pifont}
\usepackage{etoolbox}
\usepackage[normalem]{ulem}
\usepackage{booktabs}
\usepackage{longtable}
\usepackage{array}
\usepackage[capitalise,noabbrev]{cleveref}
\usepackage{cite}
\usepackage{multibib}
\usepackage{url}
\usepackage{algorithm}
\usepackage{algpseudocode}
\usepackage{paralist}
\usepackage{mathrsfs}
\usepackage{relsize}
\usepackage{stmaryrd}
\usepackage{multirow}
\usepackage[lambda,n,operators]{cryptocode}

\newtoggle{notes}
\toggletrue{notes} 

\newcommand{\F}{\mathbb{F}}

\makeatletter
\DeclareFontFamily{OMX}{MnSymbolE}{}
\DeclareSymbolFont{MnLargeSymbols}{OMX}{MnSymbolE}{m}{n}
\SetSymbolFont{MnLargeSymbols}{bold}{OMX}{MnSymbolE}{b}{n}
\DeclareFontShape{OMX}{MnSymbolE}{m}{n}{
    <-6>  MnSymbolE5
   <6-7>  MnSymbolE6
   <7-8>  MnSymbolE7
   <8-9>  MnSymbolE8
   <9-10> MnSymbolE9
  <10-12> MnSymbolE10
  <12->   MnSymbolE12
}{}
\DeclareFontShape{OMX}{MnSymbolE}{b}{n}{
    <-6>  MnSymbolE-Bold5
   <6-7>  MnSymbolE-Bold6
   <7-8>  MnSymbolE-Bold7
   <8-9>  MnSymbolE-Bold8
   <9-10> MnSymbolE-Bold9
  <10-12> MnSymbolE-Bold10
  <12->   MnSymbolE-Bold12
}{}

\let\llangle\@undefined
\let\rrangle\@undefined
\DeclareMathDelimiter{\llangle}{\mathopen}%
                     {MnLargeSymbols}{'164}{MnLargeSymbols}{'164}
\DeclareMathDelimiter{\rrangle}{\mathclose}%
                     {MnLargeSymbols}{'171}{MnLargeSymbols}{'171}
\makeatother

\newif\ifdraft
\ifnum\allow=1 \drafttrue \else \draftfalse \fi

\providecolor{Maroon}{HTML}{9E1B32}
\providecolor{Peach}{HTML}{E58F65}
\providecolor{SkyBlue}{HTML}{6FA8DC}    
\providecolor{Bittersweet}{HTML}{FE6F5E}
\providecolor{BenHuangGreen}{HTML}{2E8B57}
\providecolor{Bartosz}{HTML}{AAAC11}
\providecolor{FrancescoBlue}{HTML}{345995}
\providecolor{PierreLucOrange}{HTML}{FCA503}

\ifdraft
  \newcommand{\fixme}[1]{\noindent\colorbox{yellow}{\scriptsize\textbf{FIXME}}\,%
      \textbf{\textcolor{Maroon}{#1}}}
  \newcommand{\todo}[1]{\noindent\colorbox{yellow}{\scriptsize\textbf{TODO}}\,%
      \textbf{\textcolor{Maroon}{#1}}}
  \newcommand{\authornote}[3]{\colorbox{#2}{\scriptsize\textbf{#1}}\,%
      {\textcolor{#2}{#3}}}
  \newcommand{\manuel}[1]{\authornote{MANUEL}{Peach}{#1}}
  \newcommand{\jieyi}[1]{\authornote{JIEYI}{SkyBlue}{#1}}

  \newcommand{\benhuang}[1]{\authornote{BENHUANG}{BenHuangGreen}{#1}}
  
  \newcommand{\francesco}[1]{\authornote{FRANCESCO}{FrancescoBlue}{#1}}
  \newcommand{\pierreluc}[1]{\authornote{PIERRELUC}{PierreLucOrange}{#1}}

\else
  \newcommand{\fixme}[1]{}\newcommand{\todo}[1]{}
  \newcommand{\manuel}[1]{}\newcommand{\jieyi}[1]{}
  \newcommand{\benhuang}[1]{}
  \newcommand{\francesco}[1]{}
  \newcommand{\pierreluc}[1]{}
\fi

\providecommand{\ket}[1]{\ensuremath{\left\lvert #1\right\rangle}}

\providecommand{\challengename}{ECDSA.Fail}

\definecolor{arInk}{HTML}{22262E}   \definecolor{arMuted}{HTML}{6B7280}
\definecolor{arAccent}{HTML}{3D5A80} \definecolor{arAccentBg}{HTML}{E9EEF5}
\definecolor{arWarm}{HTML}{B07D3B}   \definecolor{arWarmBg}{HTML}{F7EEDE}
\definecolor{arBorder}{HTML}{C9CED6} \definecolor{arEdge}{HTML}{7A828E}

\newcommand{\freshBaseQ}{1{,}338}
\newcommand{\freshBaseTM}{1.128}

\newcommand{\freshBaseAny}{39}

\newcommand{\freshConservativeQ}{1{,}392}
\newcommand{\freshConservativeTM}{1.253}

\newcommand{\freshConservativeAny}{0}

\newcommand{\freshJumpQ}{1{,}162}
\newcommand{\freshJumpTM}{1.523}

\newcommand{\freshJumpAny}{215}

\title{Efficient Record-and-Replay Arithmetic for Quantum Elliptic-Curve Point Addition}
\titlerunning{Record-and-Replay Arithmetic for Quantum Point Addition}
\author{%
\small
Jieyi Long\inst{1}\emailicon{jieyi@thetalabs.org}, Theodore Pender\inst{2},
Zhao Huang\inst{3}, Manuel B. Santos\inst{4},
Samrendra Kumar Singh\inst{5}, Bartosz Naskr\k{e}cki\inst{6,7},
BitWonka\inst{8},
Pierre-Luc Dallaire-Demers\inst{9}, Francesco Giannicola,
Ruben M. L. Paschoarelli\inst{4}, Oli Freuler\inst{5},
Jackie Chia-Hsun Lee\inst{10}, Vasily Gnuchev, Gopi Kannappan, John Boyer,
Xavier Butler, Akash Balasubramani\inst{5}, Jordan Newman,
Bereket Dereje,
Alexander Hertlein\inst{11}, Robert Kodra\inst{2},
Lucas Levy\inst{5}, Shaan Patel, JT Rose, Matt Zweil,
Okechukwu Wisdom\inst{12}, Tarek El-Eter\inst{2}, Edison Lee, Michael Dong\inst{3},
Alan Li\inst{3}, Anto Joseph\inst{13}, Duy Nguyen\inst{13},
Other ECDSA.Fail Leaderboard Contributors%
\thanks{See Appendix~\ref{app:credit-contributors} for the complete contributor list.}, 
Gajesh Naik\inst{13}\emailicon{gajesh@eigenlabs.org},
Gautham Anant\inst{13}\emailicon{gautham@eigenlabs.org},
Soubhik Deb\inst{13}\emailicon{soubhik@eigenlabs.org},
Justin Drake\inst{14}\emailicon{justin@ethereum.org}%
}

\authorrunning{J. Long et al.}

\institute{%
\small
\textsuperscript{1}Theta Labs \quad
\textsuperscript{2}Starknet Foundation \quad
\textsuperscript{3}Brevis \quad
\textsuperscript{4}MultiVM Labs \\
\textsuperscript{5}StarkWare \quad
\textsuperscript{6}Adam Mickiewicz University Pozna\'{n} \quad
\textsuperscript{7}Warsaw University of Technology \quad
\textsuperscript{8}Octav \quad
\textsuperscript{9}Pauli Group \quad
\textsuperscript{10}ScienceVR \\
\textsuperscript{11}Sei Labs \quad
\textsuperscript{12}Stanford Free Systems Lab \quad
\textsuperscript{13}Eigen Labs \quad
\textsuperscript{14}Ethereum Foundation%
}
\date{September 2026}
\hypersetup{
  hidelinks
}

\begin{document}
\raggedbottom

\maketitle
\vspace{-0.3cm}

\begin{abstract}
We study reversible \texttt{secp256k1} point-addition circuits developed through {\challengename} for Shor's elliptic-curve discrete-logarithm algorithm. Two complementary constructions improve record-and-replay GCD algorithm: Jump-2 groups binary-GCD steps and compresses their decisions using base-$5$ encoding, while ping-pong uses fixed register alternation and one-bit decisions to avoid full-register-width comparisons and data-dependent swaps. Fused replay combines doubling and signed addition into one modular correction. Both constructions support quantum-addressed window selection with measurement-based lookup cleanup. We compare three circuits on $100{,}000$ fresh inputs across nine lookup-table configurations. A separately tested repair uses $1{,}419$ qubits and $1.356$ million mean executed Toffolis, with no detected failures on another $100{,}000$ inputs. We also provide conditional coherent-error analysis and reversible safegcd comparisons. Under the stated window allowance, repaired-circuit resources lie below Google's low-gate caps \cite{google26} and Schrottenloher's low-gate estimates \cite{s26}, but differing accounting and correctness evidence preclude formal dominance. The results concern individual window-selected additions, not complete Shor computations.


\end{abstract}

\ifnum\cameraversion=0
  \begingroup
  \makeatletter
  \let\l@title\@gobbletwo
  \let\l@author\@gobbletwo
  \makeatother
  \endgroup
\fi

\section{Introduction}
\label{sec:intro}

This paper studies circuits developed through {\challengename}\footnote{Challenge website: \url{https://ecdsa.fail}. Source repository: \url{https://github.com/Layr-Labs/ecdsafail-challenge}.}, a public research challenge for reducing the quantum resources required for elliptic-curve point addition on \texttt{secp256k1} via collaboration between humans and AI agents. Point addition is performed repeatedly in Shor's algorithm for the elliptic-curve discrete logarithm problem and accounts for much of its arithmetic cost~\cite{shor97,proos-zalka03}. 

Our effort focuses on modular inversion, the most expensive component of the point-addition circuit. We implement inversion using a reversible extended Euclidean algorithm (EEA), which combines a greatest common divisor (GCD) computation with coefficient updates. An efficient approach for modular inversion is the record-and-replay GCD~\cite{ksgz25,s26}, where the circuit first records the decisions made during the GCD computation in a quantum transcript and then reuses that transcript to update the field-valued registers. We study two complementary ways to reduce the resulting resource cost. The Jump-2 construction combines consecutive binary-GCD rounds and compresses their control information, reducing the number of transcript qubits. The comparison-free \emph{ping-pong} construction takes a different approach: it alternates the updated register according to a fixed schedule and records one bit indicating whether each round adds or subtracts the other operand. This produces a longer transcript but eliminates full-width qubit register comparisons and data-dependent swaps, which simplifies both the Euclidean updates and the subsequent field arithmetic. We further reduce replay cost by combining doubling and signed addition into a single modular operation.


\subsection{Main Contributions}
\label{sec:contributions}

\begin{itemize}
\item \textbf{Compressed record-and-replay GCD.} We describe and evaluate the Jump-2 construction, which combines consecutive binary-GCD operations into larger reversible steps and compresses their control decisions using a base-$5$ transcript encoding (\Cref{sec:jump-two-gcd}). Its $261$-step Jump-2 GCD schedule requires $609$ transcript qubits, compared with storing every binary decision separately. Together with in-place register reuse and specialized arithmetic for \texttt{secp256k1}, this construction gives the historical mixed-addition operating point of $1{,}151$ logical qubits and $1{,}299{,}453$ mean executed Toffolis.
\item \textbf{A comparison-free ping-pong GCD.} As an alternative to compressing the decisions of a comparison-based GCD computation, we introduce the ping-pong construction, which removes full-width qubit register comparisons and data-dependent swaps altogether. The circuit alternates the register updated in each round according to a fixed schedule and records one bit indicating whether the round adds or subtracts the other operand. We prove that each round is reversible and preserves odd operands (\Cref{sec:ping-pong-dialog-gcd}). During replay, the recorded decisions control the corresponding updates of the field-valued registers. We reduce this cost by combining modular doubling and signed addition into a single replay operation (\Cref{sec:fused-modular-replay}). At matched correction width, an isolated inverse-replay cell uses $393$ rather than $446$ static Toffolis, with one additional workspace qubit. In a controlled backend comparison that fixes the Karatsuba square and disables post-emission optimization passes, ping-pong saves $340{,}442$ static Toffolis while increasing the circuit width from $1{,}150$ to $1{,}321$ qubits.
\item \textbf{Adaptation for quantum-addressed windowed point addition.} With modest additional overhead, we were able to adapt both Jump-2 and ping-pong to an interface in which a quantum address selects an elliptic-curve point from a classically precomputed table, similar to the settings of prior work~\cite{s26,google26}. Measurement-based uncomputation clears the temporary lookup data while preserving coherence~\cite{gidney19windowed}. This cleanup reduces the mean executed Toffoli count by $9.56\%$ for Jump-2 and $12.50\%$ for ping-pong without increasing peak width (\Cref{sec:windowed-addition-compatible-circuit}).
\item \textbf{Validation, targeted repair, and scope.} The three-circuit frozen study tests $100{,}000$ fresh inputs across nine lookup-table configurations. A separate targeted repair removes the tested zero-payload and zero-slope failures at $1{,}419$ qubits and $1.356$ million mean executed Toffolis, with no detected failures on another frozen $100{,}000$-input study. Explicit supported inputs still violate the round and width schedules. We provide a conditional coherent-error composition bound and executable full-width safegcd references to distinguish these empirical operating points from stronger correctness guarantees (\Cref{sec:targeted-repair-evaluation,sec:state-dependent-composition,app:safegcd-reference}).
\end{itemize}

Detailed resource accounting shows that the savings come primarily from simplifying the Euclidean value updates and the field arithmetic performed during replay (\Cref{sec:mechanism-effects}). After adapting the circuit to quantum-addressed window selection and using more conservative arithmetic parameters, ping-pong reaches a favorable resource operating point relative to previously published constructions. Under the leading-order window allowance used in the external studies, its reported qubit width and counted work are below those of Schrottenloher's low-gate construction and Google's published low-gate caps (\Cref{sec:external-pa-comparison}). These comparisons provide useful context but do not establish formal dominance because the circuits use different interfaces, correctness evidence, and resource-accounting conventions.


The companion challenge report~\cite{long26ecdsafail} describes the broader research process and leaderboard results. This paper focuses on the circuit mechanisms, controlled comparisons, and correctness analysis, including the targeted repair and reference constructions evaluated here. The observational challenge record does not isolate the causal contribution of AI agents (\Cref{sec:challenge-and-benchmark}).

\subsection{Related Work}
\label{sec:related-work}

Proos and Zalka developed detailed quantum circuits for prime-field ECDLP, while Kaye and Zalka studied space-efficient elliptic-curve arithmetic over binary fields~\cite{proos-zalka03,kaye-zalka05}. Subsequent work improved reversible arithmetic and windowed implementations~\cite{rnsl17,hjn20,litinski23}. Babbush et al.\ report point-addition resource bounds supported by zero-knowledge proofs, while Schrottenloher provides an explicit circuit construction~\cite{google26,s26}. The record-and-replay Euclidean architecture was introduced by Khattar and Gidney et al.~\cite{ksgz25} and later adapted to ECDLP point addition by Schrottenloher~\cite{s26}. 

Low-bit plus-minus GCD is an established classical technique~\cite{bigou-tisserand15}. Safegcd also avoids full magnitude comparisons by using parity, an auxiliary state $\delta$, and conditional signed swaps~\cite[Section~8]{bernstein-yang19}. Ping-pong differs by keeping both operands odd, fixing which register is updated in each round, and recording a single sign decision. The contribution is this organization of reversible control and field-arithmetic replay. Safegcd's termination bound does not apply to ping-pong. Our full-width references use shared payload primitives (\Cref{app:safegcd-reference}), but do not establish a resource advantage over optimized reversible safegcd.

The coefficient identities used by ping-pong apply over odd prime fields, while its inexpensive modular corrections additionally exploit the sparse modulus of \texttt{secp256k1}. Karatsuba squaring and measurement-based lookup cleanup are established techniques~\cite{parent2017karatsuba,gidney19windowed}. The research process is also related to generate-and-evaluate approaches to scientific search~\cite{alphatensor-quantum,funsearch,alphaevolve}.

\section{Preliminaries}
\label{sec:prelims}

\subsection{Reversible Quantum Circuits}
\label{sec:qc-background}

Quantum arithmetic must be implemented reversibly: a circuit must preserve enough information to reconstruct its input and return temporary workspace to its initial state. This requirement often introduces additional registers and operations compared with ordinary classical arithmetic. Some temporary values can instead be cleared using measurement-based uncomputation~\cite{bennett73,gidney18}. Because such cleanup must also remove any input-dependent phase, producing the correct classical output alone does not establish that the circuit behaves correctly within a larger quantum computation.

\subsection{Elliptic-Curve Arithmetic}
\label{sec:ecc-quantum}

We use \texttt{secp256k1}, $E:y^2=x^3+7$ over $\F_p$, with $p=2^{256}-2^{32}-977$~\cite{sec2}. Its group has prime order $r$ and generator $G$. Throughout, $n=256$ is the field bit-width. Let $R=(x_R,y_R)$ be the accumulator, $A=(x_A,y_A)$ the addend, and $R'=R+A$. For finite points with $x_R\neq x_A$, the affine formulas are
\begin{equation}
  \begin{aligned}
    \lambda&=(y_R-y_A)(x_R-x_A)^{-1},\\
    x_{R'}&=\lambda^2-x_R-x_A,\\
    y_{R'}&=\lambda(x_A-x_{R'})-y_A.
  \end{aligned}
  \label{eq:group-law}
\end{equation}
All coordinate arithmetic is modulo $p$. ~\cref{alg:point-addition} provides the algorithm for in-place addition of two elliptic curve points using a quantum circuit. The in-place constructions additionally require $x_A\neq x_{R'}$, so their generic domain is $R,A\neq\mathcal O$ and $R\notin\{A,-A,-2A\}$.

Define $d_x=x_R-x_A$ and $d_y=y_R-y_A$. The slope can be recovered from either side of the addition~\cite{proos-zalka03}:
\begin{equation}
  \lambda=d_y d_x^{-1}
  =\frac{y_{R'}+y_A}{x_A-x_{R'}}\pmod p.
  \label{eq:output-slope}
\end{equation}
This identity allows the circuit to reconstruct the slope from the output coordinates and uncompute the associated temporary values without retaining both input coordinates. It is therefore a key ingredient in implementing point addition in place.


\begin{algorithm}[t]
  \caption{In-place generic addition of a classically specified addend $A$ to a quantum accumulator $R$, realizing $R'=R+A$. Registers $X$ and $Y$ hold field elements and are recycled in place.}
  \label{alg:point-addition}
  \begin{algorithmic}[1]
    \Statex \textbf{Input:} $R=(x_R,y_R)\in E(\F_p)$ and $A=(x_A,y_A)\in E(\F_p)$.
    \Statex \textbf{Assumptions:} $R,A\neq\mathcal O$ and $R\notin\{A,-A,-2A\}$; equivalently, $x_R\neq x_A$ and $x_A\neq x_{R'}$, where $R'=R+A$.
    \Statex \textbf{Output:} The coordinate registers contain $R'=(x_{R'},y_{R'})=R+A$.
    \State $X \gets x_R-x_A$;\quad $Y \gets y_R-y_A$
      \Comment{$X=d_x,\;Y=d_y$}
    \State $Y \gets YX^{-1}$
      \Comment{$Y=\lambda$}
    \State $X \gets X+3x_A$
      \Comment{$X=x_R+2x_A$}
    \State $X \gets X-Y^2$
      \Comment{$X=x_A-x_{R'}$}
    \State $Y \gets YX$
      \Comment{$Y=\lambda(x_A-x_{R'})$}
    \State $Y \gets Y-y_A$
      \Comment{$Y=y_{R'}$}
    \State $X \gets x_A-X$
      \Comment{$X=x_{R'}$}
  \end{algorithmic}
\end{algorithm}

\subsection{Windowed Point-Addition Interfaces}
\label{sec:windowed-point-addition-interfaces}

For $P=[k]G$, Shor's algorithm recovers $k$ by Fourier sampling
\begin{equation}
  f(u,v)=[u]G+[v]P=[u+kv]G,
  \label{eq:probe}
\end{equation}
whose period direction is $(-k,1)$~\cite{shor97,proos-zalka03}. Its arithmetic consists of two scalar multiplications. A semiclassical Fourier implementation reuses control qubits, making the arithmetic workspace a principal storage cost.

For either $(s,B)=(u,G)$ or $(v,P)$, a bitwise implementation of $[s]B$ adds $A_i=[2^i]B$ under control of $s_i$:
\begin{equation}
  U_{A_i}:\ket{s_i}\ket{R}\longmapsto\ket{s_i}\ket{R+[s_i]A_i}.
  \label{eq:mixed-add}
\end{equation} 

Windowing instead groups scalar bits into $w$-bit words~\cite{s26, google26}. For a window beginning at bit $i$, a classical table $\mathcal T_B^{(i)}[j]=[j2^i]B$ defines
\begin{equation}
  U_{\mathcal T_B^{(i)}}:\ket{j}\ket{R}
  \longmapsto\ket{j}\ket{R+\mathcal T_B^{(i)}[j]}.
  \label{eq:windowed-add}
\end{equation}
A QROM lookup prepares $\ket{j}\ket{\mathcal T_B^{(i)}[j]}$. The selected coordinates are quantum data and must be uncomputed after use. At $n=256$ and $w=16$, the two scalar multiplications contain $32$ conceptual additions. Direct initialization and classical postprocessing reduce the schedule used in prior resource estimates to $28$ calls~\cite{google26,litinski23}. These schedule savings do not establish correctness of a particular approximate implementation.

\section{The {\challengename} Benchmark and Open Autoresearch}
\label{sec:challenge-and-benchmark}

\subsection{Point-Addition Task and Resource Model}
\label{sec:task}
\label{sec:cost}

The benchmark used by {\challengename} supplies a quantum accumulator in two $256$-qubit registers and a classical addend $A$, and targets
\begin{equation}
  U_A:\ket{x_R}\ket{y_R}
  \longmapsto\ket{x_{R'}}\ket{y_{R'}},
  \qquad R'=R+A.
  \label{eq:task}
\end{equation}
As required by reversibility, after the computation finishes, temporary registers must return to zero. The constructions use the generic-affine domain of \Cref{sec:ecc-quantum}. The evaluator filters identity inputs and $R=\pm A$, but does not separately filter $R=-2A$. Under random-scalar sampling this additional exception has probability approximately $1/r$. It remains outside the arithmetic correctness claim even if absent from a finite test.

Circuits use the restricted kickmix model~\cite{google26}, including reversible permutation gates, diagonal phase gates, reset, and destructive $X$-basis measurement with classical feed-forward. Unrestricted Hadamard gates and persistent $\ket{+}$ preparation are excluded. This permits efficient basis-input simulation of the arithmetic and its phase and ancilla checks, rather than general quantum computation~\cite{aharonov03}.

Peak logical width $Q$ is the maximum number of simultaneously live qubits. Average executed Toffoli count $T$ is the empirical mean over the evaluation inputs, counting each executed CCX or CCZ as one Toffoli-equivalent operation. It is not the static gate count or circuit depth. The primary objective is
\begin{equation}
  S=Q\times T.
  \label{eq:score}
\end{equation}
This spacetime-inspired score captures a storage--work trade-off but omits routing, parallelism, error correction, and other architecture-dependent costs. Additional leaderboard tracks explored alternative resource objectives but are not analyzed further here, and the observed operating points do not establish an optimal trade-off.

\subsection{Validation and Approximation}
\label{sec:evaluator}

The evaluator derives $9{,}024$ pseudorandom inputs from a SHAKE256 hash of the serialized circuit~\cite{google26}. Each input is checked for correct affine output (\texttt{cls}), clean ancillas (\texttt{anc}), and absence of unwanted relative phase (\texttt{pha}). Passing these checks is finite-sample evidence, not an all-input proof.

The implementations permit a small error rate through finite iteration schedules, narrowed comparisons, and truncated arithmetic, similar to prior work~\cite{s26,google26}. These choices are distinct from exact transformations under stated preconditions. Moreover, adding identity operations can change the circuit hash without changing the arithmetic. Searching such nonces can select a favorable validation set and alter both acceptance and the reported mean cost. We therefore separate source changes and static counts from changes in sampled execution.

Test inputs generated independently of the circuit hash complement the official evaluation. The first shared-corpus comparison is exploratory because its inputs had previously been studied (\Cref{sec:cleanup-ablation}). The subsequent independent validation study freezes its circuits and analysis before generating fresh inputs (\Cref{sec:bounded-validation}). The quantity $Q\times T/\hat p$, where $\hat p$ is an empirical success probability, is used only as a per-call sensitivity proxy. It does not describe retrying an addition within a coherent Shor computation.

\subsection{Collaborative Circuit Optimization}
\label{sec:autoresearch-paradigm}

\textit{Open Autoresearch} coordinates independent human--agent teams through a shared evaluator, repository, and leaderboard \cite{liu25autoresearch, karpathy26autoresearch, hyperspace26agi, long26ecdsafail}. {\challengename} is an example of Open Autoresearch, where the first eight weeks involved more than one hundred participants and four hundred promoted submissions, exposing source, measurements, contributor attribution, and declared AI models. Teams used isolated worktrees, staged tests, and separate output, phase, and ancilla diagnostics. Public changes spread across teams without requiring all experiments to follow one branch.

Human research remained essential. For example, the challenge adopted the record-and-replay architecture introduced by Khattar and Gidney et al.~\cite{ksgz25} after Schrottenloher published his adaptation of the algorithm to ECDLP~\cite{s26}. Participants selected their own tools and budgets, and many attempts remained private. This observational record does not isolate an agent advantage. Detailed workflows and attribution remain in the \href{https://github.com/jieyilong/ecdsafail-autoresearch-harness}{contributed harness repository} and the long-form account.

The remainder of the paper describes the circuit optimizations developed through this Open Autoresearch process, their adaptation to windowed Shor, and their empirical evaluation. \textbf{\Cref{sec:gate-efficient-circuits}} presents the in-place point-addition architecture and its principal optimizations, including Jump-2 transcript compression, the comparison-free ping-pong construction, fused replay, and specialized modular arithmetic. \textbf{\Cref{sec:windowed-addition-compatible-circuit}} adapts Jump-2 and ping-pong to quantum-addressed windowed addition and explains how measurement-based uncomputation clears temporary lookup data. \textbf{\Cref{sec:results}} evaluates the resulting circuits through controlled ablations, source-level resource accounting, and frozen validation on fresh inputs, and compares their operating points with prior work. 

\section{Gate-Efficient Point-Addition Circuits}
\label{sec:gate-efficient-circuits}

We first describe two implementations of the same in-place point-addition algorithm. The Jump-2 construction reduces transcript storage through base-$5$ encoding. The comparison-free ping-pong construction instead stores a longer transcript whose symbols control simpler arithmetic. While both constructions build on the record-and-replay dialog-GCD architecture~\cite{ksgz25,s26}, they improve it in different ways. Jump-2 directly extends the original design by grouping consecutive binary-GCD operations into larger steps and compressing their control decisions, whereas ping-pong reorganizes the Euclidean control flow itself by replacing full-width qubit register comparisons and data-dependent swaps with fixed register alternation and a single sign bit per ordinary round.

\subsection{In-Place Point Addition and Dialog-GCD}
\label{sec:best-QxT-scoring-circuit}

\Cref{fig:best-qxt-circuit-pipeline} describes the in-place computation of point addition $R'=R+A$ following the procedure described in~\Cref{alg:point-addition}. After forming $d_x=x_R-x_A$ and $d_y=y_R-y_A$, the circuit replaces $Y=d_y$ by the slope $\lambda=d_y d_x^{-1}$. The next two updates give $X=d_x+3x_A-\lambda^2=x_A-x_{R'}$. Multiplication then produces $Y=\lambda(x_A-x_{R'})=y_{R'}+y_A$, allowing both output coordinates to be recovered in place. The addend $A$ is supplied classically in the benchmark and is materialized only during the coordinate updates that require it.


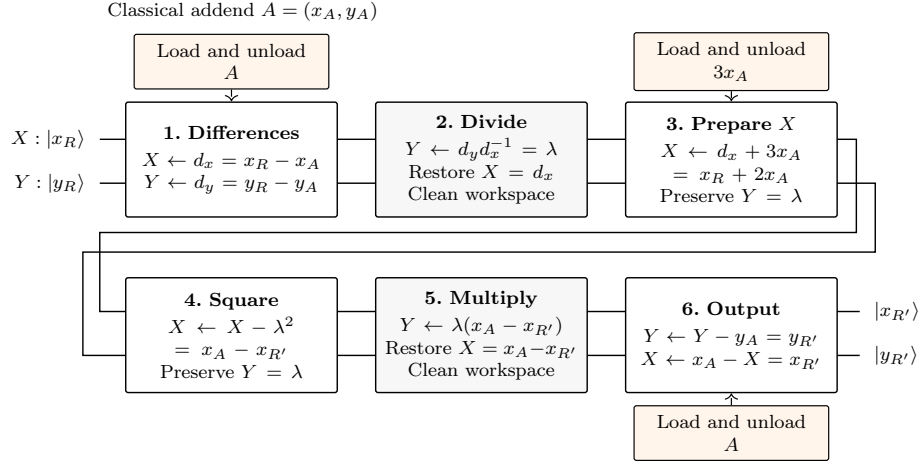
\begin{figure}[tbp]
  \centering
  \resizebox{\linewidth}{!}{%
  \begin{tikzpicture}[
    x=1cm,y=1cm,
    wire/.style={line width=0.6pt},
    stage/.style={draw,line width=0.6pt,fill=white,rounded corners=1pt,
      text width=2.85cm,minimum height=1.7cm,inner sep=4pt,align=center,font=\fontsize{8}{9.5}\selectfont},
    gcdstage/.style={stage,fill=black!3},
    access/.style={draw,line width=0.5pt,fill=orange!8,rounded corners=1pt,
      text width=2.6cm,inner sep=4pt,align=center,font=\fontsize{8}{9.5}\selectfont},
    lab/.style={font=\fontsize{8}{9.5}\selectfont},
    accessflow/.style={->,dashed,line width=0.5pt}
  ]
    \node[lab,anchor=west] at (0,6.0) {Classical addend $A=(x_A,y_A)$};
    \draw[wire] (0,4.10) -- (11.20,4.10) -- (11.20,2.72)
      -- (0,2.72) -- (0,1.55) -- (11.25,1.55);
    \draw[wire] (0,3.45) -- (11.48,3.45) -- (11.48,2.43)
      -- (-0.25,2.43) -- (-0.25,0.90) -- (11.25,0.90);
    \node[stage] (s1) at (1.95,3.80)
      {\textbf{1. Differences}\\[2pt]
       $X\gets d_x=x_R-x_A$\\
       $Y\gets d_y=y_R-y_A$};
    \node[gcdstage] (s2) at (5.65,3.80)
      {\textbf{2. Divide}\\[2pt]
       $Y\gets d_y d_x^{-1}=\lambda$\\
       Restore $X=d_x$\\
       Clean workspace};
    \node[stage] (s3) at (9.35,3.80)
      {\textbf{3. Prepare $X$}\\[2pt]
       $X\gets d_x+3x_A$\\
       $=x_R+2x_A$\\
       Preserve $Y=\lambda$};
    \node[stage] (s4) at (1.95,1.20)
      {\textbf{4. Square}\\[2pt]
       $X\gets X-\lambda^2$\\
       $=x_A-x_{R'}$\\
       Preserve $Y=\lambda$};
    \node[gcdstage] (s5) at (5.65,1.20)
      {\textbf{5. Multiply}\\[2pt]
       $Y\gets\lambda(x_A-x_{R'})$\\
       Restore $X=x_A-x_{R'}$\\
       Clean workspace};
    \node[stage] (s6) at (9.35,1.20)
      {\textbf{6. Output}\\[2pt]
       $Y\gets Y-y_A=y_{R'}$\\
       $X\gets x_A-X=x_{R'}$};
    \node[access] (a1) at (1.95,5.25) {Load and unload\\$A$};
    \node[access] (a3) at (9.35,5.25) {Load and unload\\$3x_A$};
    \node[access] (a6) at (9.35,-0.25) {Load and unload\\$A$};
    \draw[accessflow] (a1.south) -- (s1.north);
    \draw[accessflow] (a3.south) -- (s3.north);
    \draw[accessflow] (a6.north) -- (s6.south);
    \node[lab,anchor=east] at (-0.10,4.10) {$X:\ket{x_R}$};
    \node[lab,anchor=east] at (-0.10,3.45) {$Y:\ket{y_R}$};
    \node[lab,anchor=west] at (11.35,1.55) {$\ket{x_{R'}}$};
    \node[lab,anchor=west] at (11.35,0.90) {$\ket{y_{R'}}$};
  \end{tikzpicture}%
  }
  \caption{Six-stage pipeline of the Jump-2 mixed-addition circuit \texttt{8e9c9a2}. The two wires are the recycled coordinate registers. Classical addend data is materialized only where needed. The two shaded stages share dialog-GCD machinery but use different denominators. The same point-addition identities underlie ping-pong. Arithmetic is modulo $p$ on the generic domain $R,A\neq\mathcal O$, $R\notin\{A,-A,-2A\}$.}
  \label{fig:best-qxt-circuit-pipeline}
\end{figure}

Among these steps, the slope calculation $\lambda=d_y d_x^{-1}$ requires the modular inverse of $d_x$. This inverse can be computed using the extended Euclidean algorithm applied to $p$ and the canonical integer representative of $d_x$. A direct reversible implementation would update two remainders together with the full-width B\'ezout coefficients used to recover the inverse. Because the inputs may be in superposition, comparisons, swaps, and arithmetic branches must be performed coherently and later uncomputed. Keeping the remainders, coefficients, and control information live at the same time therefore requires substantial quantum workspace.

Dialog-GCD reduces this cost by separating the computation into a record phase and a replay phase~\cite{ksgz25,s26}. As shown in \Cref{alg:dialog-gcd-record}, the record phase applies a binary Euclidean reduction to an operand pair $(u,v)$ initialized as $(p,x)$. In each round, the circuit first records whether $v$ is odd in a bit $b_i$. If $v$ is odd and $u>v$, a second bit $s_i$ records that the operands must be swapped. The circuit then conditionally swaps the operands, subtracts $u$ from $v$ when $b_i=1$, and divides $v$ by two. The subtraction makes $v$ even, so the final division is exact. The two-bit symbols $\sigma_i=(b_i,s_i)$ form the transcript $\tau$. A fixed round budget makes the schedule independent of the input, with inactive rounds padded after the pair reaches $(1,0)$.

\begin{algorithm}[t]
  \caption{Dialog-GCD record phase, based on Schrottenloher~\cite[Alg.~2]{s26}.}
  \label{alg:dialog-gcd-record}
  \begin{algorithmic}[1]
    \Statex \textbf{Input:} $(u,v)=(p,x)$ with $1\leq x<p$; fixed round budget $L_{\mathrm{div}}$
    \Statex \textbf{Output:} coherent transcript $\tau=(\sigma_0,\ldots,\sigma_{L_{\mathrm{div}}-1})$
    \State $\tau\gets()$
    \For{$i=0,\ldots,L_{\mathrm{div}}-1$}
      \State $b_i\gets[v\text{ is odd}]$
        \Comment{subtraction flag}
      \State $s_i\gets b_i\wedge[u>v]$
        \Comment{pre-subtraction swap flag}
      \State $\sigma_i\gets(b_i,s_i)$;\quad
             $\tau\gets\tau\mathbin{\Vert}\sigma_i$
      \If{$s_i$}
        \State $(u,v)\gets(v,u)$
      \EndIf
      \If{$b_i$}
        \State $v\gets v-u$
      \EndIf
      \State $v\gets v/2$
    \EndFor
    \Statex \textbf{Successful endpoint:} $(u,v)=(1,0)$
  \end{algorithmic}
\end{algorithm}

For example, consider inversion of $x=d_x=5$ modulo $p=13$. In the first round, $v=5$ is odd and $u=13>5$, so the algorithm records $\sigma_0=(1,1)$, swaps the operands, subtracts, and halves:
\begin{equation}
  (13,5)\longmapsto(5,13)\longmapsto(5,8)
  \longmapsto(5,4).
\end{equation}
Applying the same rules in the following rounds gives
\begin{equation}
  (13,5)\longrightarrow(5,4)\longrightarrow(5,2)
  \longrightarrow(5,1)\longrightarrow(1,2)
  \longrightarrow(1,1)\longrightarrow(1,0),
  \label{eq:dialog-gcd-example}
\end{equation}
and produces the active transcript
\begin{equation}
  \tau=\bigl((1,1),(0,0),(0,0),(1,1),(0,0),(1,0)\bigr).
  \label{eq:dialog-gcd-example-transcript}
\end{equation}
For instance, the symbols $(0,0)$ correspond to rounds that only halve an even second operand, while the final symbol $(1,0)$ records subtraction without a preceding swap. Replaying this transcript on $(0,1)$ produces $(8,0)$. Its first component is the desired inverse because $5^{-1}=8\pmod{13}$. The transcript therefore retains the control information needed to recover the inverse without keeping full-width B\'ezout coefficients throughout the Euclidean reduction.

More generally, the division and multiplication stages (Stages 2 and 5 in~\Cref{fig:best-qxt-circuit-pipeline}) use the same dialog-GCD machinery. For a nonzero field element $x$, the swaps, subtractions, and divisions by two specified by $\tau$ define an invertible linear map $D_x$ over $\F_p$. For a successful path, $D_x(p,x)=(1,0)$. Since $p=0$ in $\F_p$, this also gives $D_x(0,x)=(1,0)$. By linearity,
\begin{equation}
  D_x(0,y)=(yx^{-1},0),
  \qquad
  D_x^{-1}(y,0)=(0,yx).
  \label{eq:jump-two-linear}
\end{equation}
The first identity implements division by replaying the transcript, while the second implements multiplication by reversing the same operations. The circuit then reverses the Euclidean record phase to restore the operand register and clear the transcript and its workspace.

\subsection{Jump-2 Steps and Base-5 Transcript Encoding}
\label{sec:jump-two-gcd}
\label{sec:base-five-codec}

An ordinary dialog-GCD step performs one mandatory halving. If the resulting operand remains even, another step is needed to halve it again. Jump-2 includes a possible second halving in the same recorded step, allowing the Euclidean reduction to make more progress per transcript symbol.

The preceding example with $p=13$ and $x=d_x=5$ illustrates this opportunity. An ordinary dialog-GCD step swaps the operands, subtracts, and performs the mandatory halving:
\begin{equation}
  (13,5)
  \xrightarrow{\mathrm{swap,\ subtract}}
  (5,8)
  \xrightarrow{/2}
  (5,4).
  \label{eq:ordinary-dialog-example}
\end{equation}
Because the second operand remains even, ordinary dialog-GCD would require another step to halve it again. Jump-2 instead includes this second halving in the same step:
\begin{equation}
  (13,5)
  \xrightarrow{\mathrm{swap,\ subtract}}
  (5,8)
  \xrightarrow{/2}
  (5,4)
  \xrightarrow{/2\;(s_2=1)}
  (5,2).
  \label{eq:jump-two-example}
\end{equation}
The bit $s_2$ records whether the additional halving occurred. By reducing the number of GCD steps along many Euclidean paths, this construction can shorten the transcript and lower the number of qubits required to store it.

Formally, let $(u,v)$ denote the Euclidean operand pair, rather than the Shor exponent registers. After a special first iteration that handles the unknown parity of $x$, every regular Jump-2 iteration begins with $u$ odd and $v$ even. It first halves $v$ and performs a second halving when the first quotient remains even:
\begin{equation}
  v\longmapsto
  \begin{cases}
    v/2, & v/2\text{ is odd},\\
    v/4, & v/2\text{ is even}.
  \end{cases}
  \label{eq:jump-two-core}
\end{equation}
If the resulting $v$ is odd, the circuit compares the operands, swaps them when $v<u$, and subtracts $u$ from $v$. Since both operands are odd, their difference is even and establishes the boundary for the next regular iteration. No further halving is performed in the current iteration. If $v$ remains even after the halving phase, no subtraction or swap is needed.

Each regular Jump-2 step records a symbol $\sigma_i=(b_i,s_i,s_{2,i})$. The bit $s_{2,i}$ is one exactly when $v$ remains even after the mandatory first halving, in which case the circuit halves it again. After the halving phase, $b_i$ records whether the resulting operand is odd and therefore requires subtraction, while $s_i$ records whether the operands are swapped before that subtraction. If $s_{2,i}=0$, the result of the first halving is odd, so $b_i=1$. If $s_{2,i}=1$, the result after the second halving may be either even or odd. Finally, a swap can occur only when subtraction is active, so $s_i=1$ implies $b_i=1$. These constraints leave five reachable symbols:
\begin{equation}
  \mathcal S=\{(0,0,1),(1,0,0),(1,1,0),(1,0,1),(1,1,1)\}.
  \label{eq:jump-two-alphabet}
\end{equation}

A bijection $\phi:\mathcal S\to\{0,1,2,3,4\}$ encodes each reachable symbol as a base-$5$ digit. Three consecutive Jump-2 GCD steps fit into seven qubits because
\begin{equation}
  z=d_0+5d_1+25d_2<125<2^7.
  \label{eq:base-five-value}
\end{equation}
The resulting reversible encoding is
\begin{equation}
  \ket{\sigma_0}_3\ket{\sigma_1}_3\ket{\sigma_2}_3
  \longmapsto
  \ket{z}_7\ket{0}^{\otimes2}.
  \label{eq:base-five-unitary}
\end{equation}
The injective assignment on reachable symbols extends to a permutation of the full nine-qubit computational basis. It returns two qubits to zero, and its inverse recovers the three symbols during reverse replay.

The $261$-step schedule encodes the special initial step in two qubits. Its remaining $260$ symbols are stored in $86$ seven-qubit blocks containing three symbols each and one five-qubit block containing the final two symbols. The complete transcript therefore occupies
\begin{equation}
  2+86\mathbin{\times}7+5=609
\end{equation}
qubits. Encoding uses workspace released as the active Euclidean operands shrink, and reverse replay decodes and clears one block at a time. The base-$5$ codec is exact on the reachable alphabet. Approximation enters through the finite step budget and the truncated arithmetic used elsewhere in the implementation.

\subsection{Comparison-Free Ping-Pong Dialog-GCD}
\label{sec:ping-pong-dialog-gcd}

The Jump-2 construction reduces transcript storage, but its Euclidean value updates retain the decision logic of ordinary dialog-GCD. A typical round compares the two full-width operands, conditionally swaps them so that subtraction proceeds in the required order, performs the subtraction, and then divides the updated operand by two. In a reversible circuit, the comparison result must be computed and later uncomputed, while the controlled swap acts across the full active width of both operands. These operations account for a substantial part of the Toffoli cost.


The ping-pong GCD addresses these costs through a fixed-schedule Euclidean computation that avoids full-width operand comparisons and data-dependent swaps. It combines low-bit plus-minus GCD methods~\cite{bigou-tisserand15} with dialog replay~\cite{ksgz25,s26}, with the pseudocode provided in \Cref{alg:ping-pong-dialog-gcd}. Rather than comparing the operands' magnitudes, the circuit examines only their two least significant bits. If the operands have the same residue modulo four, it adds them; otherwise, it subtracts them. In either case, the resulting numerator is congruent to two modulo four, so dividing it by two produces another odd operand. 

For example, still consider our running example where $p=13$ and $x=d_x=5$. Starting from $(13,5)$, both operands $13\equiv5\equiv1\pmod4$, so the second operand is updated to $(5+13)/2=9$, giving $(13,9)$ and transcript bit $e=0$. The target then alternates to the first operand. Since $13$ and $9$ again have the same residue, the next update gives $(13+9)/2=11$ and the pair $(11,9)$. In the following round, $11\equiv3\pmod4$ and $9\equiv1\pmod4$, so the circuit subtracts and obtains $(9-11)/2=-1$, giving $(11,-1)$ and recording $e=1$. Continuing the fixed alternation reaches $(1,-1)$ with transcript $(0,0,1,0,1,0,0)$. Thus, one bit per round records whether the update used addition or subtraction, making the computation reversible without storing comparison or swap decisions.

\begin{algorithm}[t]
  \caption{Comparison-free ping-pong GCD}
  \label{alg:ping-pong-dialog-gcd}
  \begin{algorithmic}[1]
    \Statex \textbf{Input:} nonzero denominator $x\in\F_p$; fixed round budget $L$
    \Statex \textbf{Output:} terminal signed pair and transcript $\tau$
    \State $\rho_0\gets p$;\quad
           $\rho_1\gets x$ if $x$ is odd, otherwise
           $\rho_1\gets x-p$
    \State $\tau\gets()$
    \For{$i=0,\ldots,L-1$}
      \If{$i$ is even}
        \State $(s,t)\gets(\rho_0,\rho_1)$
          \Comment{update $\rho_1$}
      \Else
        \State $(s,t)\gets(\rho_1,\rho_0)$
          \Comment{update $\rho_0$}
      \EndIf
      \State $e_i\gets t[1]\oplus s[1]$
      \State $t\gets(t+(-1)^{e_i}s)/2$
      \State append $e_i$ to $\tau$
    \EndFor
    \Statex \textbf{Successful endpoint:}
      $(\rho_0,\rho_1)=(\eta_0,\eta_1)$ with
      $\eta_0,\eta_1\in\{-1,+1\}$
  \end{algorithmic}
\end{algorithm}

Ping-pong trades a longer transcript for cheaper value updates, simpler transcript decoding, and less expensive replay arithmetic. Here, \emph{comparison-free} refers specifically to the Euclidean value walk. The modular replay arithmetic still contains bounded carry and overflow tests.

For the canonical representative $1\leq x<p$, \Cref{alg:ping-pong-dialog-gcd} initializes the two operand registers as $(\rho_0,\rho_1)=(p,x^\star)$, where
\begin{equation}
  x^\star=
  \begin{cases}
    x, & x\text{ is odd},\\
    x-p, & x\text{ is even}.
  \end{cases}
  \label{eq:ping-pong-odd-lift}
\end{equation}
Both entries are odd and $x^\star\equiv x\pmod p$. The update schedule is fixed: in even-numbered rounds, $\rho_0$ is the source $s$ and $\rho_1$ is the target $t$, while in odd-numbered rounds their roles are reversed. The source and target therefore bounce between the two qubit registers from one round to the next, giving the method its \emph{ping-pong} name. For an odd signed integer $z$, let $z[1]$ denote bit one of its two's-complement representation, with bit zero the least significant bit. Thus, $z[1]=0$ when $z\equiv1\pmod4$ and $z[1]=1$ when $z\equiv3\pmod4$. In round $i$, the update rule is
\begin{equation}
  e_i=t[1]\oplus s[1],
  \qquad
  t'=\frac{t+(-1)^{e_i}s}{2}.
  \label{eq:ping-pong-forward}
\end{equation}

\begin{lemma}[Odd-preserving reversible value update]
\label{lem:ping-pong-value-update}
For odd signed integers $s,t$, \Cref{eq:ping-pong-forward} produces an odd integer $t'$ with $\gcd(s,t')=\gcd(s,t)$. Retaining $e$ makes the map $(s,t)\mapsto(s,t',e)$ injective. Its inverse on the image is $t=2t'-(-1)^e s$.
\end{lemma}
\begin{proof}
When the residues agree, their sum is $2\pmod4$. When they differ, their difference is $2\pmod4$. Therefore $t'$ is odd. Since $s$ is odd,
\begin{equation}
  \gcd(s,t)=\gcd(s,t\pm s)=\gcd\!\left(s,\frac{t\pm s}{2}\right)
  \label{eq:ping-pong-gcd-invariant}
\end{equation}
where the chosen sign is determined by $e$. Substitution gives the stated inverse, so distinct input pairs cannot have the same output triple.
\end{proof}

The target alternates between the two physical registers according to a fixed schedule, so selecting it requires no quantum comparison. For coprime inputs, the Euclidean walk eventually reaches a terminal pair $(\eta_0,\eta_1)$ with $\eta_0,\eta_1\in\{-1,1\}$. Subsequent rounds leave this pair unchanged. However, different components of a superposition may converge at different times, so the input-dependent convergence time is not available as a classical stopping condition. The implementation therefore uses a fixed budget of $704$ rounds and stores one transcript qubit per round. This is longer than the $609$-qubit Jump-2 transcript but avoids multi-bit symbol decoding.

For the linear map $D_x$ determined by the recorded path, replay on a payload $(0,y)$ gives signed copies of the quotient:
\begin{equation}
  D_x(0,y)=yx^{-1}(\eta_0,\eta_1)\pmod p.
  \label{eq:ping-pong-division}
\end{equation}
After correcting the terminal signs, both coefficient registers encode $yx^{-1}$. One copy can then be cleared by a reversible XOR, provided that the two registers have identical bitwise encodings. Equality only modulo $p$ is not sufficient for this cleanup. Running the same transcript in reverse implements multiplication:
\begin{equation}
  D_x^{-1}(\eta_0y,\eta_1y)=(0,yx)\pmod p.
  \label{eq:ping-pong-multiplication}
\end{equation}
At the level of one reverse replay cell, the target coefficient is updated as
\begin{equation}
  c_t=2c_t'-(-1)^e c_s\pmod p.
  \label{eq:ping-pong-reverse-payload}
\end{equation}
Reversing the complete Euclidean computation also restores the operand registers and clears the transcript. The first transition and the initial replay operation are specialized, while the remaining rounds use the fused cells described below. Here, \emph{comparison-free} refers to the Euclidean value walk, not to every modular-arithmetic or phase-correction operation.

\begin{proposition}[Conditional arithmetic correctness]
\label{prop:ping-pong-correctness}
Suppose the exact recurrence on $(p,x^\star)$ reaches a signed-unit pair within $L$ rounds. If the value operations, modular replay, and phase corrections are correct, with canonical coefficient encodings at cleanup, the recorded path implements division and multiplication as in \Cref{eq:ping-pong-division,eq:ping-pong-multiplication}, while restoring the denominator and clearing the transcript.
\end{proposition}
\begin{proof}
Each recorded round is invertible and preserves the GCD. Conditioned on its transcript, the composition is linear over $\F_p$. Since $(p,x^\star)\equiv(0,x)\pmod p$, scaling by $yx^{-1}$ gives the division identity. Applying the inverse composition to $y(\eta_0,\eta_1)$ gives the multiplication identity. Canonical encodings justify the bitwise duplicate and zero-register cleanup. Reversing the value rounds restores their inputs and clears the recorded decisions. Correct phase corrections preserve these maps coherently.
\end{proof}

The inequality $|t'|\leq(|t|+|s|)/2$ alone does not establish convergence. \Cref{lem:ping-pong-termination} gives a conservative polynomial termination bound for the exact recurrence by combining average steps with magnitude-reducing difference steps. That bound is much larger than $704$ and does not justify the submitted shrinking-width schedule. The reference diagnostics find $33$ of $99{,}998$ nonidentity inputs with an intended denominator walk requiring more than $704$ rounds. Eight also violate the width schedule, and the observed maxima are $725$ and $723$ rounds for division and multiplication. The diagnostics are associated with $32$ of the $40$ nonidentity mixed-circuit failures, but do not trace the first divergent gate (see the \href{https://github.com/jieyilong/ecdsafail-circuit-evidence/blob/d6a93d1baaa2cc472278d2b419ef41d2d4708074/sources/frozen-tools/windowed_pingpong/EXPERIMENTS.md#exact-recurrence-diagnostics}{recorded recurrence diagnostics}). Proposition~\ref{prop:ping-pong-correctness} therefore remains conditional for the measured implementation.

\subsubsection{Fused Modular Replay}
\label{sec:fused-modular-replay}

Stage~5 of the point-addition schematic in \Cref{fig:best-qxt-circuit-pipeline} multiplies the slope $\lambda$ by $x_A-x_{R'}$ to obtain $y_{R'}+y_A$. Ping-pong implements this multiplication by recording the Euclidean path for the denominator $x_A-x_{R'}$ and then replaying that path in reverse on the coefficient registers. Writing the source coefficient as $y$, the target coefficient as $z$, and $\sigma=1-e$, each ordinary reverse-replay round performs
\begin{equation}
  z\gets2z+(-1)^\sigma y\pmod p.
\end{equation}
The same replay machinery is used in the forward direction in Stage~2 to compute $\lambda=d_y d_x^{-1}$, using the transcript recorded for $d_x$.

A direct implementation performs the reverse-replay update in two stages. It first doubles $z$ and applies a modular correction, then performs the signed addition and applies another correction.  Fused replay retains the overflow from doubling and the carry from addition, then uses them to perform one combined correction. It therefore removes one arithmetic pass, although the phase correction needed to clear a measured carry remains.

Let $B_0=2^n$ and $F=B_0-p$. For \texttt{secp256k1}, $F=2^{32}+977$. Suppose the $n$-bit low word produced by the doubling and signed addition is $w$. The discarded multiples of $B_0$ are summarized by $\kappa\in\{-1,0,1,2\}$, determined from the overflow of $2z$ and the carry or borrow of the signed addition. Since $B_0=p+F$,
\begin{equation}
  2z+(-1)^\sigma y
  =w+\kappa B_0
  \equiv w+\kappa F\pmod p.
  \label{eq:fused-correction}
\end{equation}
An unfused cell applies a correction after doubling and another after adding or subtracting $y$. The fused cell combines these corrections. The equation above uses the ordinary low-word representation; the emitted subtraction branch instead complements the target and reverses the correction sign before restoring that representation. Because $F$ is sparse, the selected constant can be supplied bit by bit without allocating a full correction register. The retained doubling overflow is erased using the corrected target parity, source parity, sign, and addition carry. The addition carry is measured and phase-corrected. Forward replay similarly combines signed addition and halving.

The submitted implementation applies the correction only within a $56$-bit window. It is correct when no carry or borrow escapes this window and the carry-erasure comparisons recover their predicates correctly. The full-width algebra is exact, but these truncated folds and comparisons remain approximation sources.

At the matched $56$-bit fold width, fusion reduces one inverse-replay cell from $446$ to $393$ static Toffolis, an $11.9\%$ reduction, while increasing local workspace from $87$ to $88$ qubits. The source-default unfused cell uses $512$ Toffolis and $89$ workspace qubits, but this comparison also includes a wider correction range. The historical $83{,}484$-Toffoli mean ablation contains additional implementation differences and should not be attributed to fusion alone. Detailed derivations and source-level accounting appear in \Cref{app:replay-cell-accounting}.

\subsubsection{Targeted Zero-Payload Repair}
\label{sec:zero-payload-repair}

The original replay can represent field zero by the word $p$ and leave input-dependent phases on zero-slope inputs. We therefore test a targeted variant that records $h=[c\neq0]$ for the input payload $c$. Forward replay uses the masked sign $he_i$, reverse replay uses $h(1-e_i)$, and terminal sign corrections are also masked by $h$. On the zero branch, both coefficient registers stay at canonical zero. On the nonzero branch, the sign choices are unchanged. The flag is uncomputed from the output zero predicate. This cleanup is valid when the implemented input and output words agree on whether they are zero, as correct canonical multiplication or division by a nonzero field element guarantees. It remains conditional when other arithmetic fails. A blocked zero test limits temporary storage, and the persistent flag costs one qubit.

The variant also uses the exact chunk-carry predicate $[z<a]\lor(c_{\mathrm{in}}\land[z=a])$, with the incoming carry kept live until its phase correction is complete, and canonical subtraction for the two input differences and final $y$ update. It retains the remaining guarded replay arithmetic and uses $768$ value rounds. These changes remove the tested zero-payload and zero-slope defects without replacing every replay cell by the canonical reference. They do not repair all arithmetic boundaries or justify the shrinking value widths. In particular, explicit field denominators exceed $768$ exact rounds (\Cref{app:round-counterexamples}). The integrated measurements and remaining failures are reported in \Cref{sec:targeted-repair-evaluation}.

\subsection{Specialized Modular Arithmetic}
\label{sec:specialized-squaring}
\label{sec:consprop}
\label{sec:dead-code-elimination}

The Jump-2 point-addition circuit specializes its square-subtract operation to $p=2^{256}-F$, where $F=2^{32}+977$. Let $\beta=2^{128}$ and split the slope as $\lambda=\lambda_{\mathrm{lo}}+\beta\lambda_{\mathrm{hi}}$. Define $a_2=\lambda_{\mathrm{lo}}^2$, $b_2=\lambda_{\mathrm{hi}}^2$, and $c_2=(\lambda_{\mathrm{lo}}+\lambda_{\mathrm{hi}})^2$. The standard Karatsuba identity~\cite{parent2017karatsuba} gives
\begin{equation}
  \begin{aligned}
    \lambda^2&=a_2+\beta(c_2-a_2-b_2)+\beta^2b_2\\
    &\equiv a_2+\beta(c_2-a_2-b_2)+Fb_2\pmod p.
  \end{aligned}
  \label{eq:secp-karatsuba-square}
\end{equation}
Instead of retaining a full $512$-bit product, the circuit computes each of $c_2,a_2,b_2$ in turn, accumulates its contribution into $X$, and uncomputes it before constructing the next product. The respective updates are $-\beta c_2$, $-a_2+\beta a_2$, and $\beta b_2-Fb_2$. Their sum is $-\lambda^2$ modulo $p$.

Each smaller square forms an off-diagonal bit product only once and doubles its contribution. Multiplication by $F$ uses $F=2^{32}+2^{10}-2^5-2^4+1$. The algebra is exact, while truncated corrections require separate justification.

Ping-pong uses a different product-register squaring implementation, retaining product workspace during its modular reduction. The square is therefore part of the full-circuit difference from Jump-2, not a controlled ablation of the Euclidean recurrence alone. Its coherent identity control is described in \Cref{sec:windowed-addition-compatible-circuit}.

The coordinate updates also load classical inputs only when needed and uncompute carries by measurement~\cite{gidney18}. Quantum-selected addends must retain coherent conditional arithmetic (\Cref{sec:windowed-addition-compatible-circuit}).


\section{Quantum-Addressed Windowed Point Addition}
\label{sec:windowed-addition-compatible-circuit}

As reviewed in \Cref{sec:windowed-point-addition-interfaces}, windowed Shor uses a quantum window register to select an addend coherently from a classically precomputed table. The {\challengename} benchmark instead supplies the addend $A$ classically and therefore omits this lookup/use/unlookup interface. To bridge this gap, we adapt both Jump-2 and ping-pong to quantum-addressed classical tables~\cite{google26,s26}. We also replace reverse lookup replay with measurement-based cleanup while holding the underlying arithmetic settings fixed. The resulting circuits implement the single-call point-addition interface required by windowed Shor~\cite{google26,s26}.

\subsection{Interface and Temporary Storage}

Let $\mathcal T$ be a classically precomputed table with $2^w$ entries and $J$ its $w$-qubit address register. The window-selected circuit implements
\begin{equation}
  U_{\mathcal T}:\ket{j}_J\ket{R}_{XY}\ket{0}_W
  \longmapsto
  \ket{j}_J\ket{R+\mathcal T[j]}_{XY}\ket{0}_W.
  \label{eq:windowed-kernel-map}
\end{equation}
The specification applies on the supported affine domain and includes the identity entry. The table consists of $2\cdot2^w$ classical $n$-bit coordinate words. These values control the table-dependent gates, while $J$ selects a row coherently. The topology is independent of the table contents, so shifted tables of multiples of $G$ or $P$ can be substituted without recompilation.

The selected coordinates are loaded only when forming the input differences, preparing $X$, and recovering the output coordinates in \Cref{fig:best-qxt-circuit-pipeline}. They are coherently unloaded before the peak-width GCD and squaring stages. Stage~3 loads $x_A$, adds $x_A$ and $2x_A$ to $X$, reverses the temporary doubling, and unloads $x_A$, avoiding a separate table column for $3x_A$.

Row zero stores $(0,0)$ as the encoding of $\mathcal O$. A coherent flag $b=[j\neq0]$ disables the square-subtract operation and final modular negation when this row is selected. The division and multiplication stages still run, but cancel under the stated arithmetic assumptions. For ping-pong's product-register square, the wrapper computes $M=b\lambda$, squares it, and uncomputes it. This adds $512$ CCX gates, while capped carry workspace prevents the mask from increasing the arithmetic peak or narrowing another arithmetic range.

Although the selected point occupies two temporary $256$-bit coordinate registers, these payloads are unloaded before the arithmetic peak and therefore do not add $512$ qubits to it. At $w=16$, the persistent window overhead consists of the $16$ address qubits and the flag $b$. At fixed arithmetic settings, these controls increase original ping-pong from $1{,}321$ to $1{,}338$ qubits and the later targeted repair from $1{,}402$ to $1{,}419$, both by $17$. The repair's larger starting width comes from its longer GCD schedule and additional arithmetic state, not the lookup interface. Jump-2 has the same $17$ controls but reaches its windowed peak with six fewer non-window qubits live, giving a net increase from $1{,}151$ to $1{,}162$.



\subsection{Exact Split-Address Cleanup}
\label{sec:split-address-cleanup}

A QROM lookup temporarily entangles the address register with the selected table word:
\begin{equation}
  \sum_j\alpha_j\ket{j}\ket{0}
  \longmapsto
  \sum_j\alpha_j\ket{j}\ket{c_j}.
\end{equation}
After the point-addition stage has used and restored $c_j$, this temporary payload must be erased without revealing the address or disturbing its superposition. Reversing the complete lookup accomplishes this but requires another expensive QROM traversal. We instead use measurement-based uncomputation~\cite{gidney19windowed}.

Let $m\in\{0,1\}^d$ be the result of measuring the restored $d$-bit payload in the $X$ basis. The measurement releases the payload register and leaves only an address-dependent phase:
\begin{equation}
  \sum_j\alpha_j\ket{j}\ket{c_j}
  \longmapsto
  2^{-d/2}\sum_j(-1)^{m\cdot c_j}\alpha_j\ket{j}\ket{0},
  \label{eq:payload-measurement}
\end{equation}
where the inner product is modulo two. The measurement outcomes are uniformly distributed and do not reveal $j$. Since both $m$ and the table entries $c_j$ are classical, the phase $(-1)^{m\cdot c_j}$ is known and can be canceled exactly. Measurement therefore introduces no additional approximation when the phase correction is applied correctly.

To apply this correction efficiently, split the $w$-bit address register into low and high halves and decode each half into a one-hot register. For every table row whose phase must be flipped, a CZ gate connects the corresponding low and high flags. Only the pair associated with the selected address is active, so these gates apply exactly the required phase. The two one-hot decoders are then reversed and cleared.

\begin{figure}[tbp]
\centering
\begin{tikzpicture}[x=1cm,y=1cm,
  block/.style={draw,line width=0.5pt,text width=2.25cm,minimum height=1.65cm,align=center,inner sep=3pt,font=\small},
  flow/.style={->,line width=0.6pt}]
  \node[block] (m) at (0,0) {Measure payload\\$X$ basis, outcome $m$};
  \node[block] (d) at (2.9,0) {Decode address\\Low and high halves};
  \node[block] (z) at (5.8,0) {Correct phase\\$(-1)^{m\cdot c_j}$};
  \node[block] (u) at (8.7,0) {Clear decoders\\Unitary inverse};
  \draw[flow] (m.east) -- (d.west);
  \draw[flow] (d.east) -- (z.west);
  \draw[flow] (z.east) -- (u.west);
  \draw[flow,dashed] (m.north) -- (0,1.2) -- (5.8,1.2) -- (z.north);
  \node[font=\footnotesize,fill=white,inner sep=2pt] at (2.9,1.2) {Classical measurement record};
  \node[font=\footnotesize] at (4.35,-1.2) {The quantum address is preserved. Released payload space hosts the decoders.};
\end{tikzpicture}
\caption{Conceptual exact-cleanup sequence, not the full point-addition wiring. The payload must first be restored to its table value. Corrected measurement branches have input-independent magnitude, and the decoder workspace returns to zero.}
\label{fig:exact-qrom-cleanup}
\end{figure}

The two coordinate columns share the same decoders. Stage~3 uses a specialized $K=2$ lookup for the single coordinate $x_A$. Its controlled swaps are reversed before the restored payload banks are measured, after which the remaining phase depends only on the upper $w-1$ address bits.

A one-hot decoder for $k$ address bits uses $2^k-1$ Toffolis, and reversing it has the same cost. The complete split-address correction therefore costs
\begin{equation}
  C(w)=2\left(2^{\lfloor w/2\rfloor}+2^{\lceil w/2\rceil}\right)-4
  \label{eq:split-cleanup-cost}
\end{equation}
Toffolis. At $w=16$, each unload costs $C(16)=1{,}020$ Toffolis. The specialized $K=2$ unload also costs $1{,}020$, including its reversed swaps. Across the three unloads, the cost falls from $164{,}090$ to $3{,}060$ Toffolis, saving $161{,}030$.

The decoder registers are allocated only after the measured payload has been released. Their largest local allocation is approximately $1{,}041$ qubits, below the arithmetic peaks of both Jump-2 and ping-pong. Split-address cleanup therefore reduces the Toffoli count without increasing peak width $Q$.

\subsection{Conditional Coherence}
\label{sec:lookup-coherence}

Measurement-based cleanup must remove both temporary values and the phases introduced by measuring them. For example, measuring a temporary AND $c=ab$ in the $X$ basis contributes the phase $(-1)^{mab}$ when the outcome is $m$. Applying $\mathrm{CZ}_{a,b}^{m}$ cancels this phase~\cite{gidney18}. The required controls must remain available until the correction is complete.

\begin{proposition}[Conditional coherence of the lookup interface]
\label{prop:lookup-coherence}
Let $\mathcal D$ be the set of valid address--accumulator pairs $(j,R)$, and define the corresponding input subspace
\begin{equation}
  \mathcal H_{\mathrm{valid}}
  =
  \operatorname{span}\left\{
    \ket{j}_J\ket{R}_{XY}\ket{0}_W:
    (j,R)\in\mathcal D
  \right\}.
\end{equation}
Assume that the point-addition arithmetic acts coherently and with the correct relative phases on this subspace, restores each lookup payload to its table value before cleanup, and returns all other workspace qubits to $\ket{0}$. In particular, its corrected measurement branches must differ from the intended arithmetic map only by input-independent scalar factors. If every $X$-basis measurement of a payload or routing ancilla is followed by the required classically controlled phase correction, then the quantum channel implemented by the wrapper agrees with
\begin{equation}
  \rho\longmapsto U_{\mathcal T}\rho U_{\mathcal T}^{\dagger}
\end{equation}
for every state $\rho$ supported on $\mathcal H_{\mathrm{valid}}$.
\end{proposition}

\begin{proof}
Let $K_{\mathbf m}$ be the Kraus operator associated with the complete measurement record $\mathbf m$. Before correction, the $X$-basis measurements contribute known phases that may depend on the address and payload values. The payload correction cancels the phase in \Cref{eq:payload-measurement}, and the routing corrections cancel the phases introduced by measured ancillas. The remaining branch amplitude depends only on $\mathbf m$, not on the input state. Therefore, for every $\ket{\psi}\in\mathcal H_{\mathrm{valid}}$,
\begin{equation}
  K_{\mathbf m}\ket{\psi}
  =
  \gamma_{\mathbf m}U_{\mathcal T}\ket{\psi},
  \qquad
  \sum_{\mathbf m}|\gamma_{\mathbf m}|^2=1,
  \label{eq:lookup-branch-condition}
\end{equation}
where $\gamma_{\mathbf m}$ is independent of $\ket{\psi}$. Consequently, the probability $|\gamma_{\mathbf m}|^2$ of observing $\mathbf m$ is independent of the input. Discarding the classical measurement record gives
\begin{equation}
  \sum_{\mathbf m}
  K_{\mathbf m}\rho K_{\mathbf m}^{\dagger}
  =
  U_{\mathcal T}\rho U_{\mathcal T}^{\dagger},
\end{equation}
which proves the claim.
\end{proof}

The proposition assumes phase-correct arithmetic as well as correct basis outputs. A diagonal sign error can preserve every output word and still destroy interference. It does not address finite round budgets, truncated arithmetic, or exceptional inputs. The following bound states what additional evidence would suffice for composition.

\subsection{State-Dependent Error and Composition}
\label{sec:state-dependent-composition}

Let $\Pi_i$ project onto an input subspace on which every corrected branch of call $i$ agrees with its ideal isometry $U_i$, up to an input-independent scalar as in \Cref{eq:lookup-branch-condition}. Let $\rho_{i-1}^{\mathrm{ideal}}$ be the ideal state immediately before that call, including its correlations with the remaining registers, and set $q_i=\operatorname{Tr}[(I-\Pi_i)\rho_{i-1}^{\mathrm{ideal}}]$. The bad-subspace argument in \Cref{app:coherent-error} gives
\begin{equation}
  \frac12\left\|\rho_{\mathrm{out}}-\rho_{\mathrm{out}}^{\mathrm{ideal}}\right\|_1
  \leq\min\left\{1,2\sum_i\sqrt{q_i}\right\}.
  \label{eq:state-dependent-composition}
\end{equation}
No independence between calls is required. The weights must, however, refer to the actual ideal intermediate states and a certified good subspace. Random affine-input failure rates cannot simply replace them. For example, immediately after direct initialization from a $16$-bit low window, the accumulator has only $2^{16}$ possible values. It is not uniformly distributed over the curve. Thus the bound supplies a conditional route to a full-algorithm guarantee, while identifying the arithmetic and distributional obligations that remain.

\section{Correctness and Resource Evaluation}
\label{sec:results}
\label{sec:results-summary}

The preceding sections introduced the Jump-2 and ping-pong backends and their adaptation to quantum-addressed window selection. We now evaluate the resulting circuits in terms of peak logical width $Q$, mean executed Toffoli count $T$, and observed correctness. We first select conservative ping-pong parameters using $16{,}384$ mixed-addition inputs reserved for parameter tuning. After freezing the implementations, we compare windowed Jump-2, original ping-pong, and conservative ping-pong on the same $100{,}000$ fresh inputs across nine lookup-table configurations. Controlled ablations and an earlier matched corpus isolate resource effects not addressed by this fresh-input comparison. Finally, we compare the measured operating points with the published Google and Schrottenloher estimates under their respective interfaces, correctness evidence, and accounting conventions. All reported mean executed Toffoli counts include failed inputs.

\subsection{Evaluation Design}
\label{sec:headline-operating-points}

All three circuits in the primary study use the same $16$-bit quantum-addressed interface and exact split-address cleanup. Jump-2 uses its $261$-step schedule and $609$-qubit transcript. Original ping-pong uses $704$ rounds and the submitted arithmetic ranges. Conservative ping-pong uses $736$ rounds together with wider value registers, replay guards, and coordinate comparisons selected on a separate parameter-tuning test set. The term \emph{conservative} describes these wider parameters, not an exact or all-input-correct implementation.

The original three-circuit frozen study supplies the primary paired comparison. The targeted repair is evaluated separately in \Cref{sec:targeted-repair-evaluation}. Controlled backend, square, cleanup, and allocation studies answer narrower resource questions. The earlier shared corpus had already informed development and is treated as exploratory. Results from these protocols are not pooled. Tables use M for $10^6$ and B for $10^9$.

\FloatBarrier
\subsection{Parameter Selection and Validation Study}
\label{sec:bounded-validation}

As discussed in \Cref{sec:ping-pong-dialog-gcd,lem:ping-pong-termination}, the ideal ping-pong recurrence terminates, but the theoretical bound does not justify the implemented $704$-round budget or shrinking register widths. We therefore test whether additional rounds and wider arithmetic reduce observed failures, selecting a more conservative ping-pong configuration for the comparative evaluation. These changes do not establish all-input correctness. Jump-2 remains unchanged as a comparison point rather than undergoing a corresponding parameter search.

For parameter selection only, we compare five cumulative ping-pong configurations on a separate set of $16{,}384$ shared mixed-addition inputs (\Cref{tab:bounded-development}). These inputs are not part of the subsequent $100{,}000$-input evaluation set. Starting from \emph{Original}, \emph{More rounds} increases the budget to $736$. \emph{Wider values} adds $16$ bits to the value-width allowance, retaining the $259$-bit cap. \emph{Wider replay guards} enlarges the replay comparison and correction ranges. Finally, \emph{Wider coordinate check} widens the coordinate carry comparison from $19$ to $40$ bits, removing a traced phase error shared by the first four configurations. 

We call this final configuration (i.e. wider coordinate check) \emph{conservative ping-pong}. Its windowed version appears in \Cref{tab:fresh-operating-points}. It was selected before generating the independent evaluation set. Here, \emph{conservative} means wider parameters, not guaranteed correctness.

\begin{table}[tbp]
\centering\small
\setlength{\tabcolsep}{4pt}
\begin{tabular}{@{}lrrrrrrr@{}}
\toprule
Ping-pong configuration & Rounds & $Q$ & $T$ (M) & Cls & Pha & Anc & Any \\
\midrule
Original & $704$ & $1{,}321$ & $0.952714$ & $8$ & $5$ & $0$ & $10$ \\
More rounds & $736$ & $1{,}353$ & $0.975938$ & $3$ & $2$ & $0$ & $4$ \\
Wider values & $736$ & $1{,}359$ & $1.019115$ & $0$ & $1$ & $0$ & $1$ \\
Wider replay guards & $736$ & $1{,}375$ & $1.077979$ & $0$ & $1$ & $0$ & $1$ \\
Wider coordinate check & $736$ & $1{,}375$ & $1.078051$ & $0$ & $0$ & $0$ & $0$ \\
\bottomrule
\end{tabular}
\vspace{8pt}
\caption{Ping-pong parameter tuning on $16{,}384$ shared mixed-addition inputs. Changes are cumulative. Cls, Pha, and Anc count output, phase, and ancilla failures. Any counts inputs with at least one failure. These tuning results are not independent failure-rate estimates.}
\label{tab:bounded-development}
\end{table}

We then freeze all three windowed circuits, the input driver, and the analysis scripts. Unlike the mixed-addition tuning tests, this comparison uses the $16$-bit quantum-addressed interface and exact split-address cleanup. Each circuit receives the same $100{,}000$ fresh inputs, divided nearly equally among \textit{nine lookup-table configurations}. These nine configurations combine three base points $B\in\{G,P_1,P_2\}$ with three bit-position shifts $i\in\{0,128,240\}$, defining the precomputed point tables $\mathcal T_B^{(i)}[j]=[j2^i]B$. A fresh seed defines the public test points $P_1=[k_1]G$ and $P_2=[k_2]G$. Each trial samples a uniform address and $R=[a]G$ from a pseudorandom $256$-bit scalar. Identity accumulators and nonidentity addends with $R=\pm A$ are rejected. Identity addends are retained. No $R=-2A$ case occurs, although it is not separately filtered.

Conservative ping-pong has no detected failures, compared with $39$ for original ping-pong and $215$ for Jump-2. Its improvement over original ping-pong costs $54$ additional qubits and $11.12\%$ more mean Toffolis. These results compare selected implementations, not equal-error optima. The reported confidence bounds concern the sampled input distribution, not coherent channel error, and structured failures remain (\Cref{sec:boundary-mechanisms}). Complete parameter settings, sampling rules, confidence calculations, and evaluator checks are documented in the \href{https://github.com/jieyilong/ecdsafail-circuit-evidence/blob/d6a93d1baaa2cc472278d2b419ef41d2d4708074/sources/frozen-tools/bounded_qip/PROTOCOL.md}{evaluation protocol} and \Cref{app:artifacts}.

\begin{table}[tbp]
\centering\small
\setlength{\tabcolsep}{4pt}
\begin{tabular}{@{}lrrrrr@{}}
\toprule
Windowed circuit & $Q$ & $T$ (M) & $QT$ (B) & Failures & Upper (\%) \\
\midrule
Ping-pong, original & $1{,}338$ & $1.128$ & $1.5086$ & $39$ & $0.11527$ \\
Ping-pong, conservative & $1{,}392$ & $1.253$ & $1.7440$ & $0$ & $0.00300$ \\
Jump-2 & $1{,}162$ & $1.523$ & $1.7699$ & $215$ & $0.35345$ \\
\bottomrule
\end{tabular}
\vspace{8pt}
\caption{Fresh frozen study with $100{,}000$ inputs per circuit across nine lookup-table configurations. \textit{Ping-pong, conservative} refers to the \textit{wider coordinate check} configuration in~\Cref{tab:bounded-development}. All variants use exact split-address cleanup. Failures are the union of output, phase, and ancilla checks. Upper is a one-sided $95\%$ bound on the allocation-weighted failure rate, not a coherent-error bound. Means include failed cases.}
\label{tab:fresh-operating-points}
\end{table}

\FloatBarrier

\subsection{Ablation Studies of Resource Savings}
\label{sec:mechanism-effects}

We now examine where the resource savings come from. Ping-pong's lower gate count largely reflects cheaper Euclidean value updates and replay, while measurement-based lookup cleanup provides a separate saving for both backends. The remainder of this section provides further details. 

\subsubsection{Lookup cleanup}
\label{sec:cleanup-ablation}
Keeping the arithmetic and table interface fixed, replacing reverse lookup with split-address measurement-based cleanup reduces mean $T$ by $9.56\%$ for Jump-2 and $12.50\%$ for ping-pong without increasing $Q$ (\Cref{tab:exact-cleanup-effects}). Both save approximately $161{,}000$ static Toffolis. This comparison uses an earlier shared test set that informed interface development. Its results are exploratory and are not pooled with the primary validation study.

\begin{table}[tbp]
\centering\footnotesize
\setlength{\tabcolsep}{3.5pt}
\renewcommand{\arraystretch}{1.12}
\begin{tabular}{@{}lrrrr@{}}
\toprule
Circuit & $Q$ & Static Toffolis & Mean $T$ (M) & Counted Clifford (M) \\
\midrule
Jump-2 replay & $1{,}162$ & $1{,}754{,}382$ & $1.684161$ & $90.754$ \\
Jump-2 split & $1{,}162$ & $1{,}593{,}346$ & $1.523120$ & $69.044$ \\
Ping-pong replay & $1{,}338$ & $1{,}345{,}161$ & $1.288560$ & $96.481$ \\
Ping-pong split & $1{,}338$ & $1{,}184{,}131$ & $1.127527$ & $74.772$ \\
\bottomrule
\end{tabular}
\vspace{8pt}
\caption{Exact-cleanup effects at unchanged peak width. Clifford counts follow the evaluator convention, including CX, CZ, SWAP, HMR, and reset but excluding tracked X and Z operations. Static and executed counts are distinct.}
\label{tab:exact-cleanup-effects}
\end{table}

\subsubsection{GCD backend and squaring}
\label{sec:controlled-backend-ablation}

To separate backend savings from squaring effects, we test both backends with both squaring methods, keeping the surrounding point-addition operations fixed and disabling subsequent circuit-optimization passes (\Cref{tab:mechanism-ablation}). With Karatsuba squaring unchanged, replacing Jump-2 by ping-pong saves $340{,}442$ static Toffolis at a cost of $171$ additional qubits. Changing the square alone saves approximately $8{,}163$ Toffolis. The backend change therefore contributes substantially more to the gate reduction.

\begin{table}[tbp]
\centering\footnotesize
\setlength{\tabcolsep}{6pt}
\begin{tabular}{@{}llrr@{}}
\toprule
Backend & Square & $Q$ & Static Toffolis \\
\midrule
Legacy Jump-2 & Karatsuba & $1{,}150$ & $1{,}357{,}682$ \\
Legacy Jump-2 & Product register & $1{,}287$ & $1{,}349{,}527$ \\
Ping-pong & Karatsuba & $1{,}321$ & $1{,}017{,}240$ \\
Ping-pong & Product register & $1{,}321$ & $1{,}009{,}077$ \\
\bottomrule
\end{tabular}
\vspace{8pt}
\caption{Resource effects of changing the GCD backend and squaring implementation. All configurations share the same surrounding point-addition operations and use explicit parameter settings with post-emission optimization disabled. These are controlled resource comparisons, not equal-error comparisons. Diagnostic outcomes are retained in the evidence repository.}
\label{tab:mechanism-ablation}
\end{table}


\subsubsection{Gate-count breakdown and peak qubit width}
\label{sec:mechanism-accounting}

\Cref{tab:mechanism-static} breaks down the static Toffoli count of each complete windowed point-addition circuit by computational stage. \emph{Original} and \emph{Conservative} denote the windowed versions of the first and final ping-pong configurations in \Cref{tab:bounded-development}, respectively.

For example, the GCD value-update row counts $633{,}994$ Toffoli gates in Jump-2 and $395{,}368$ in original ping-pong. Each entry includes computing and reversing the Euclidean values across both the division and multiplication stages, not a single GCD round or call. Field-valued coefficient operations are counted separately in the replay rows.

The Share column compares Jump-2 with \emph{original}, not conservative, ping-pong. Their total pre-optimization difference is $410{,}567$ gates. GCD value updates save $633{,}994-395{,}368=238{,}626$ gates, contributing $58.1\%$ of that total saving. Ordinary field replay contributes another $37.3\%$. Negative shares indicate stages where original ping-pong costs more. The final three rows show the pre-optimization total, optimization adjustment, and final static count.

\begin{table}[tbp]
\centering\footnotesize
\setlength{\tabcolsep}{3pt}
\begin{tabular}{@{}lrrrr@{}}
\toprule
Computational stage & Jump-2 & Original & Conservative & Share \\
 & & (ping-pong) & (ping-pong) & (\%) \\
\midrule
GCD value updates & $633{,}994$ & $395{,}368$ & $439{,}196$ & $58.1$ \\
Replay setup and endpoints & $22{,}526$ & $828$ & $956$ & $5.3$ \\
Ordinary field replay & $704{,}939$ & $551{,}772$ & $670{,}876$ & $37.3$ \\
Transcript encoding/decoding & $6{,}378$ & $0$ & $0$ & $1.6$ \\
Squaring and slope masking & $55{,}668$ & $64{,}968$ & $64{,}968$ & $-2.3$ \\
Coordinate updates & $3{,}985$ & $3{,}987$ & $4{,}092$ & $0.0$ \\
Lookup and address control & $167{,}208$ & $167{,}208$ & $167{,}208$ & $0.0$ \\
\midrule
Total before optimization & $1{,}594{,}698$ & $1{,}184{,}131$ & $1{,}347{,}296$ & $100.0$ \\
Optimization adjustment & $-1{,}352$ & $0$ & $0$ & --- \\
Final static count & $1{,}593{,}346$ & $1{,}184{,}131$ & $1{,}347{,}296$ & --- \\
\bottomrule
\end{tabular}
\vspace{8pt}
\caption{Static CCX+CCZ counts for the three windowed circuits, including local phase corrections. Share is each stage's Jump-2 count minus its original ping-pong count, divided by the total pre-optimization difference of $410{,}567$, expressed as a percentage. Negative shares indicate additional cost in original ping-pong. Shares are rounded to one decimal place. Subsequent optimization is shown separately. These are static counts, not mean executed counts $T$.}
\label{tab:mechanism-static}
\end{table}



Peak qubit width instead counts registers that are live simultaneously. Ping-pong's replay peak occurs during initialization or seed processing (\Cref{tab:mechanism-peak}). Original windowed ping-pong also reaches $1{,}338$ qubits during squaring, while conservative ping-pong reaches $1{,}392$ during replay. These peaks are not added together. Jump-2 peaks during an inverse modular correction, so transcript length alone does not determine $Q$.

\begin{table}[tbp]
\centering\footnotesize
\setlength{\tabcolsep}{4pt}
\begin{tabular}{@{}lrrrrrr@{}}
\toprule
Ping-pong circuit & Payload & Tape & Values & Carries & Controls & $Q$ \\
\midrule
Mixed original & $512$ & $704$ & $16$ & $88$ & $1$ & $1{,}321$ \\
Windowed original & $512$ & $704$ & $16$ & $88$ & $18$ & $1{,}338$ \\
Windowed conservative & $512$ & $736$ & $22$ & $104$ & $18$ & $1{,}392$ \\
\bottomrule
\end{tabular}
\vspace{8pt}
\caption{Simultaneous allocation at the special replay peak. Controls include one arithmetic flag and, in windowed circuits, the 16 address qubits and nonzero-address flag. Counts describe allocator ownership, not a proof that every released wire is semantically clean. No maxima from different phases are added.}
\label{tab:mechanism-peak}
\end{table}

\FloatBarrier
\subsection{Analysis of Ping-Pong Correctness Repairs}
\label{sec:targeted-repair-evaluation}

We separately test whether specific known failures can be repaired while retaining ping-pong's low cost. The new variant combines the zero-payload repair of \Cref{sec:zero-payload-repair} with corrected chunk-carry cleanup, canonical coordinate subtraction, and $768$ rounds. It is distinct from the three configurations in the preceding validation study.

The windowed circuit uses $1{,}419$ qubits and $1{,}524{,}503$ static Toffolis, counted from the final emitted circuit. Its mixed-addition version uses $1{,}402$ qubits, so the window interface still adds only $17$. The increase from conservative windowed ping-pong's $1{,}392$ qubits instead comes from arithmetic changes: $32$ additional transcript qubits, six fewer value qubits at the longer schedule's endpoint, and one zero-payload flag. These changes add $27$ qubits independently of the windowing overhead.

The repaired division and multiplication circuits pass zero-payload tests on $64$ denominators, restoring the denominators and clearing temporary registers without detected phase errors. Complete point addition also passes the two previously failing supported zero-slope inputs, each tested over $32$ simulated runs, and a $64$-input test set at both $w=4$ and $w=16$. These results demonstrate a cheaper repair of those particular failures than the canonical reference. They do not establish correctness on all inputs.

After the targeted regressions, we froze the candidate source, emitted stream, input driver, and sampling protocol before drawing a new seed. We evaluate the repaired circuit on $100{,}000$ fresh inputs across nine lookup-table configurations, using the same three-base-point, three-shift construction and sampling rules as \Cref{sec:bounded-validation}, but newly generated public test points and input seeds. It detects no output, phase, or ancilla failures, including on four identity-addend rows. The mean executed count is $1{,}356{,}324$ Toffolis, giving $Q\times T\approx1.925$ billion. Under independent sampling, the zero count gives a one-sided $95\%$ upper bound of $0.00300\%$ on the allocation-weighted failure rate. An independent Python affine oracle checks every expected output, and all case indices and batch totals reconcile. This cohort is not paired with the original three-circuit study. It supports the measured operating point but does not supersede the structured counterexamples or establish a coherent error bound.


 
\FloatBarrier

\FloatBarrier

\subsection{Comparison with Prior Work}
\label{sec:external-pa-comparison}

For leading-order window-model context, Google and Schrottenloher add $16$ address qubits and $3\cdot2^{16}=196{,}608$ lookup operations to their reported arithmetic resources~\cite[Appendix~A.3]{google26}\cite[Tables~1--2]{s26}. Applying that allowance gives the external model points in \Cref{tab:external-pa-comparison}. Small unlookup costs are not fully itemized, so these are not measured all-in kernels. Our rows instead include the implemented lookup, cleanup, address register, and failed inputs; no additional allowance is applied.

Under this allowance, conservative ping-pong uses $70$ fewer qubits and approximately $39\%$ less counted work than Schrottenloher's low-gate construction. It also uses $49$ fewer qubits and approximately $45\%$ less counted work than Google's low-gate caps. The Schrottenloher comparison is measured versus constructed, while the Google comparison is against published upper bounds rather than undisclosed actual costs. Different correctness evidence and counting conventions preclude a formal dominance claim.

The targeted zero-payload repair also remains below the two low-gate model points in width and counted work, at $1{,}419$ qubits and $1.356$ million mean Toffolis. Its improved targeted behavior does not remove the schedule counterexamples, so the same correctness and accounting qualifications apply.

\begin{table}[!htbp]
\centering\footnotesize
\setlength{\tabcolsep}{3pt}
\renewcommand{\arraystretch}{1.12}
\begin{tabular}{@{}>{\raggedright\arraybackslash}p{3.25cm}r>{\raggedright\arraybackslash}p{1.35cm}>{\raggedright\arraybackslash}p{4.0cm}@{}}
\toprule
Circuit & $Q$ & Work (M) & Correctness evidence \\
\midrule
\multicolumn{4}{@{}l}{\textit{Published resources with leading-order window allowance}}\\
Google low-$Q$~\cite{google26} & $\leq1{,}191$ & $\leq2.897$ & ZK-attested $0/9{,}024$ FS \\
Google low-gate~\cite{google26} & $\leq1{,}441$ & $\leq2.297$ & ZK-attested $0/9{,}024$ FS \\
Schrottenloher low-$Q$~\cite{s26} & $1{,}208$ & $\approx2.59$ & Reported $10{,}000$ random successes, block-level simulation \\
Schrottenloher low-gate~\cite{s26} & $1{,}462$ & $\approx2.06$ & Same reported evidence \\
\midrule
\multicolumn{4}{@{}l}{\textit{Measured single-call window-selected kernels, $w=16$}}\\
Windowed Jump-2 & $\freshJumpQ$ & $\freshJumpTM$ & $\freshJumpAny/100{,}000$ fresh cases \\
Windowed ping-pong & $\freshBaseQ$ & $\freshBaseTM$ & $\freshBaseAny/100{,}000$ fresh cases \\
Conservative ping-pong & $\freshConservativeQ$ & $\freshConservativeTM$ & $\freshConservativeAny/100{,}000$ fresh cases \\
\midrule
\multicolumn{4}{@{}l}{\textit{Separate targeted-repair study, $w=16$}}\\
Ping-pong, zero-payload repair & $1{,}419$ & $1.356$ & $0/100{,}000$ new cases \\
\bottomrule
\end{tabular}
\vspace{8pt}
\caption{Window-model resource context under different accounting and correctness conventions. FS denotes circuit-dependent Fiat--Shamir acceptance. Google reports caps on mean executed CCX+CCZ counts. Schrottenloher reports rounded construction counts including CCX, CCZ, and AND gates. Our means count executed CCX+CCZ and include failed inputs. The rows do not form an accuracy-matched frontier.}
\label{tab:external-pa-comparison}
\end{table}

Schrottenloher reports $\epsilon\leq2^{-13.3}$ without a stated confidence level and uses block-level simulation in which some arithmetic is replaced by classical functions. Neither these simulations nor Fiat--Shamir acceptance provides uncertainty bounds directly comparable to those of our fresh-input study. The resource comparisons therefore do not establish a speedup at matched correctness guarantees. Toffoli depth, feed-forward depth, and physical implementation costs also remain unmatched.

\FloatBarrier
\subsection{Reproducibility}
\label{sec:reproducibility}

The companion evidence repository's \href{https://github.com/jieyilong/ecdsafail-circuit-evidence/releases/tag/v1.4.0}{version~1.4.0} preserves the earlier studies and adds the targeted repair, fresh validation, schedule audit, coherent-error analysis, and safegcd references in experiments~10--14~\cite{ecdsafail-evidence-v14}. A \href{https://github.com/jieyilong/ecdsafail-circuit-evidence/blob/d6a93d1baaa2cc472278d2b419ef41d2d4708074/docs/PAPER_MAP.md}{commit-pinned experiment index} maps each result to its records and scripts. Verification reconstructs the earlier $300{,}000$ outcomes and the new $100{,}000$ outcomes from the recorded results without rerunning the circuits. \Cref{app:artifacts} gives source revisions and reproduction commands.

The release includes frozen source snapshots, candidate and operation-stream hashes, fresh input/output ledgers, retained counterexamples, and reproduction scripts. The \href{https://github.com/jieyilong/ecdsafail-circuit-evidence/blob/d6a93d1baaa2cc472278d2b419ef41d2d4708074/docs/REPRODUCTION.md}{versioned reproduction guide} distinguishes record verification from source compilation and circuit execution. The new targeted-repair study remains separate from the earlier three-circuit frozen comparison.

\FloatBarrier

\section{Limitations and Outlook}
\label{sec:limitations}

\paragraph{Correctness and finite schedules.}
The ideal termination bound does not justify the measured round budgets or shrinking widths. The targeted repair removes the tested zero-payload and zero-slope failures while retaining a low resource cost, but explicit denominator and point-pair counterexamples remain. The canonical reference repairs more replay defects at higher cost and still retains the finite value schedule. Neither clean sampled outcomes nor local proofs establish correctness of the complete implementation on every input.

\paragraph{Interfaces and resources.}
The windowed kernels implement single calls, not a complete Shor computation. Testing nine lookup-table configurations does not cover every public point, shifted table, or correlated intermediate state. Initialization, exceptional additions, the full schedule, Fourier operations, and postprocessing remain to be validated together. The retry proxy is not an attack-success estimate. 
Reported counts omit depth, parallelism, memory traffic, measurement latency, and physical error correction, so they do not establish a hardware-independent optimum.

\paragraph{Attribution.}
The circuits combine prior arithmetic with implementation adaptations. The observational challenge record does not isolate a causal advantage from agents or compare models under controlled budgets. AI tools assisted with drafting, source analysis, test drivers, and the termination argument. The authors remain responsible for correctness, attribution, and the final text.

The principal open tasks are canonical replay at competitive cost, justified value widths and round tails, an optimized matched reversible safegcd baseline, and certification of the good subspaces and ideal-state weights required by the composition bound.



%
%

\bibliographystyle{splncs04}
\bibliography{refs}

%
%

\appendix
\renewcommand{\theHsection}{appendix.\Alph{section}}

\section{Contributor Statement}
\label{app:credit-contributors}

\begingroup
\small
\setlength{\emergencystretch}{2em}
\noindent The categories below are nonexclusive. Individual circuit submissions and their provenance remain recorded on the public leaderboard.

\paragraph{\textbf{Paper contributors}.} Jieyi Long led the drafting, synthesis, coordination, and manuscript-wide integration of the paper. Theodore Pender, Zhao Huang, Manuel B. Santos, and Samrendra Kumar Singh contributed major technical and empirical sections, including the quantum-circuit background, circuit-optimization analysis, results, improvement chronology, and Pareto-frontier analysis. Justin Drake, Pierre-Luc Dallaire-Demers, and Bartosz Naskr\k{e}cki contributed technical review and revision, while Francesco Giannicola contributed detailed review, reproducibility study, and editorial refinement.

\paragraph{\textbf{Challenge organizers}.} Gajesh Naik, Gautham Anant, Soubhik Deb, and Justin Drake contributed to the conception, harness setup, organization, infrastructure, administration, or supervision of the challenge.

\paragraph{\textbf{Leaderboard contributors}.} The challenge's progress reflects the cumulative contributions of participants who performed circuit research, implementation, testing, validation, and submission work. Within this broader collective effort, several key optimization techniques and plateau-breaking improvements were identified by Theodore Pender, BitWonka, Bartosz Naskr\k{e}cki, Joe Doyle, Matt Zweil, and Gajesh Naik, working in collaboration with their AI agents. These particularly consequential contributions complemented the many other advances, refinements, and validation efforts made by the wider leaderboard community. \textbf{Named leaderboard contributors}: Akash Balasubramani, BitWonka, John Boyer, Xavier Butler, Pierre-Luc Dallaire-Demers, Bereket Dereje, Mo Dong, Joe Doyle, Oli Freuler, Vasily Gnuchev, Alexander Hertlein, Zhao Huang, Anto Joseph, Robert Kodra, Edison Lee, Jackie Chia-Hsun Lee, Lucas Levy, Alan Li, Jieyi Long, Gajesh Naik, Bartosz Naskr\k{e}cki, Jordan Newman, Ruben Marcus Luz Paschoarelli, Shaan Patel, Theodore Pender, JT Rose, Manuel B. Santos, Samrendra Kumar Singh, Okechukwu Wisdom. \textbf{Other Leaderboard Contributors (Github ID)}: 0xfoobar,\allowbreak{} 0xnirlin,\allowbreak{} 0xpg,\allowbreak{} 10d9e,\allowbreak{} abipalli,\allowbreak{} aburan28,\allowbreak{} agilexmarketing,\allowbreak{} alecstecpe-oss,\allowbreak{} AnalyticETH,\allowbreak{} Angel-Arevalo,\allowbreak{} anshu4321,\allowbreak{} anupsv,\allowbreak{} austinamissah,\allowbreak{} ayushgupta0610,\allowbreak{} blocksek,\allowbreak{} borelien,\allowbreak{} bulengerk,\allowbreak{} bxue-l2,\allowbreak{} chainyoda,\allowbreak{} Chris-Moller,\allowbreak{} cryptogakusei,\allowbreak{} darius-oai,\allowbreak{} DavetheSlayer,\allowbreak{} davidxia3,\allowbreak{} dennisonbertram,\allowbreak{} eferbarn,\allowbreak{} eiso,\allowbreak{} factory-sagar,\allowbreak{} fortunexbt,\allowbreak{} Gaijin-01,\allowbreak{} gopikannappan,\allowbreak{} harrigan,\allowbreak{} hydrogenbond007,\allowbreak{} inishc,\allowbreak{} IntrepidEnki,\allowbreak{} jacksonhblau,\allowbreak{} jamesotter99,\allowbreak{} Jamgiter,\allowbreak{} Joshgriste,\allowbreak{} josusanmartin,\allowbreak{} jtaroreh,\allowbreak{} junpenglao,\allowbreak{} kamilsa,\allowbreak{} kunalvg,\allowbreak{} lucasabu1988,\allowbreak{} Luigy-Lemon,\allowbreak{} Lulu-Zhou-EigenLabs,\allowbreak{} M3kko,\allowbreak{} makimakiver,\allowbreak{} megabyte0x,\allowbreak{} mmurrs,\allowbreak{} mochimodev,\allowbreak{} mpjunior92,\allowbreak{} nandy-technologies,\allowbreak{} nikhiljha,\allowbreak{} nnn-gif,\allowbreak{} novoyd,\allowbreak{} pavvann,\allowbreak{} phileigenlabs,\allowbreak{} philuponcrypto,\allowbreak{} pq-cybarg,\allowbreak{} prasiddhnaik,\allowbreak{} pschork,\allowbreak{} ptrdsh,\allowbreak{} runeape-sats,\allowbreak{} rwnq8,\allowbreak{} Sam-Scolari,\allowbreak{} Sampriv,\allowbreak{} saucegodbased,\allowbreak{} simonik11,\allowbreak{} skb93,\allowbreak{} solimander,\allowbreak{} stevenhao,\allowbreak{} steventhornton,\allowbreak{} swadhin,\allowbreak{} Tasfia-17,\allowbreak{} tekkac,\allowbreak{} TeS44,\allowbreak{} tmid1,\allowbreak{} twangodev,\allowbreak{} twmmason,\allowbreak{} vatsalkshah,\allowbreak{} vineetguptadev,\allowbreak{} wh1sky02,\allowbreak{} xadenryan,\allowbreak{} ygboucherk,\allowbreak{} YQ-Wang,\allowbreak{} yudduy,\allowbreak{} yudongcao,\allowbreak{} zeeshan8281,\allowbreak{} zigtur,\allowbreak{} zpano,\allowbreak{} zuiris.

\endgroup

\section{Artifact Index and Reproduction}
\label{app:artifacts}

All artifacts are indexed in \href{https://github.com/jieyilong/ecdsafail-circuit-evidence/releases/tag/v1.4.0}{evidence release v1.4.0}~\cite{ecdsafail-evidence-v14}. Its \href{https://github.com/jieyilong/ecdsafail-circuit-evidence/blob/d6a93d1baaa2cc472278d2b419ef41d2d4708074/docs/PAPER_MAP.md}{commit-pinned map} links each result to frozen source, inputs, and analysis, including the separately archived shared-corpus comparisons. Circuit-source revisions and evidence-release revisions are distinct. Earlier results are preserved, not pooled with later studies.

\subsection{Validation Records}
\label{app:fresh-artifacts}

Experiment~02 contains the original three-circuit study: $100{,}000$ shared inputs and $300{,}000$ outcomes. Its freeze manifest predates input generation, and rejected-case identities were not retained. The table below reports this study only.
\begin{table}[H]
\centering\small
\setlength{\tabcolsep}{5pt}
\begin{tabular}{@{}lrrrr@{}}
\toprule
Lookup configuration & Inputs & Original PP & Conservative PP & Jump-2 \\
\midrule
$B=G,\ i=0$ & $11{,}111$ & $8$ & $0$ & $18$ \\
$B=G,\ i=128$ & $11{,}111$ & $2$ & $0$ & $20$ \\
$B=G,\ i=240$ & $11{,}111$ & $6$ & $0$ & $27$ \\
$B=P_1,\ i=0$ & $11{,}111$ & $6$ & $0$ & $26$ \\
$B=P_1,\ i=128$ & $11{,}111$ & $5$ & $0$ & $19$ \\
$B=P_1,\ i=240$ & $11{,}111$ & $2$ & $0$ & $19$ \\
$B=P_2,\ i=0$ & $11{,}111$ & $6$ & $0$ & $30$ \\
$B=P_2,\ i=128$ & $11{,}111$ & $1$ & $0$ & $31$ \\
$B=P_2,\ i=240$ & $11{,}112$ & $3$ & $0$ & $25$ \\
\bottomrule
\end{tabular}
\caption{Any-channel failures in the original study, with $\mathcal T_B^{(i)}[j]=[j2^i]B$. Each row uses the same inputs for all three circuits.}
\label{tab:fresh-strata}
\end{table}

Experiment~10 separately tests the targeted repair on $100{,}000$ fresh inputs across nine lookup-table configurations, with no detected output, phase, or ancilla failures. It uses $1{,}419$ qubits, $1{,}524{,}503$ static Toffolis, and $1{,}356{,}324.32985$ mean executed Toffolis. The matched mixed circuit uses $1{,}402$ qubits: windowing adds $17$, independently of the arithmetic repairs. Structured failures remain (\Cref{app:round-counterexamples}).

\subsection{Reproduction}

Run \texttt{python3 scripts/verify.py} from the repository root to reconstruct both studies, check hashes and diagnostics, and recompute the new study's table entries and point sums. The \texttt{-{}-oracle} option checks earlier points with OpenSSL. These record checks do not execute quantum circuits. The \href{https://github.com/jieyilong/ecdsafail-circuit-evidence/blob/d6a93d1baaa2cc472278d2b419ef41d2d4708074/docs/REPRODUCTION.md}{reproduction guide} gives separate build and execution commands, including regeneration of omitted gate streams and binaries.

\section{Termination and Local Replay Analysis}
\label{app:replay-analysis}

\subsection{A Conservative Termination Bound}

\begin{lemma}[Termination of the ordinary ping-pong recurrence]
\label{lem:ping-pong-termination}
For odd coprime signed integers $s,t$ with $\max(|s|,|t|)<2^n$, $n\geq1$, \Cref{eq:ping-pong-forward} with alternating targets reaches $|s|=|t|=1$ within $3n(n+1)$ rounds.
\end{lemma}
\begin{proof}
Write $a=|s|$, $b=|t|$. Regardless of operand signs, the updated magnitude is $(a+b)/2$ when $a\equiv b\pmod4$ and $|a-b|/2$ otherwise. Call these average and difference steps. The maximum never increases, and equal magnitudes imply the signed-unit endpoint by coprimality. Otherwise, each consecutive average step halves the positive even difference $|a-b|<2^n$, so at most $n-1$ average steps precede a difference step. Let $M,m$ be the larger and smaller magnitudes before that difference step. Its target becomes $(M-m)/2$. After the next round, which updates the other register, the maximum is at most $M/2$ if the first target was larger, and $3M/4$ otherwise. Thus at most $n+1$ rounds reduce the maximum by a factor at most $3/4$, unless termination occurs first. Three such blocks reduce it by more than half. After $3n$ blocks, a nonterminal positive maximum would be less than $2^n2^{-n}=1$, a contradiction.
\end{proof}
For $n=256$ this gives $197{,}376$ ordinary rounds over exact integers. It justifies neither the short circuit budgets nor shrinking widths; \Cref{app:round-counterexamples} gives explicit violations.

\subsection{Initialization and Cell Accounting}
\label{app:replay-cell-accounting}

For canonical $x$, set $a_0=x\bmod2$, $a_1=\lfloor x/2\rfloor\bmod2$, and $h=\lfloor x/2\rfloor$. The specialized first transition produces $h-p+a_1p+a_0(p+1)/2$, an odd integer congruent to $x/2$ modulo $p$. Its sign recovers $a_1$, while the transcript stores $a_0$. The first replay only halves, the next uses a seed cell, and ordinary inverse cells begin at index two. Hence the $704$-round circuit has $702$ ordinary inverse cells.

The \href{https://github.com/jieyilong/ecdsafail-circuit-evidence/blob/d6a93d1baaa2cc472278d2b419ef41d2d4708074/supporting/replay-analysis/README.md}{cell-accounting guide} pins source \texttt{897dda2b0cf2}, the one-Toffoli majority option, and counts before whole-circuit optimization. The common addition and phase comparisons cost $336$ Toffolis. At matched $56$-bit fold width, unfused and fused inverse cells cost $446/393$ Toffolis and $87/88$ workspace qubits, excluding payloads and sign control. The forward fused cell costs $393$ Toffolis and $87$ workspace qubits. Matching fold width does not establish identical failure sets. Exact scalar checks support the algebra, not emitted phases or representations.

\subsection{Boundary and Carry Diagnostics}
\label{sec:boundary-mechanisms}
\label{app:targeted-followup}

\href{https://github.com/jieyilong/ecdsafail-circuit-evidence/tree/d6a93d1baaa2cc472278d2b419ef41d2d4708074/experiments/07-boundary-diagnosis}{Experiment~07} separates representation and phase defects. Negation by complementation can encode zero as word $p$. Among $64$ zero-payload denominators, division/multiplication return $p$ in $28/46$ cases and raise $19/26$ phase flags, despite restored denominators and clean final ancillas. Resetting a nonzero temporary word can hide a cleanup failure while introducing phase. These are component diagnostics, not population estimates.

For an unsigned chunk with source $a$, output $z$, and incoming carry $c$, correct cleanup requires
\begin{equation}
 c_{\mathrm{out}}=[z<a]\lor(c\land[z=a]).
 \label{eq:chunk-carry-boundary}
\end{equation}
The omitted equality term causes input-dependent phase even when $0-1$ returns canonical $p-1$. Widening a strict comparison does not repair this predicate. Separate post-fold overflow and halving-underflow defects remain.

\href{https://github.com/jieyilong/ecdsafail-circuit-evidence/tree/d6a93d1baaa2cc472278d2b419ef41d2d4708074/experiments/08-followup-diagnostics}{Experiment~08} supplies two supported zero-slope witnesses: for a nontrivial cube root $\beta\in\F_p$, take $A=G$ and $R=(\beta x_G,y_G)$ or $(\beta^2x_G,y_G)$. Neither is an excluded exceptional point. At $w=4$, both original and conservative circuits return correct coordinates, but raise $27/35$ phase flags over $32$ measurement lanes per point. The $64$ lanes represent two inputs, not independent sampled pairs.

The same experiment explains smoke input $50$: its $616$- and $638$-round walks fit the original widths, but inverse-replay round $271$ loses a carry beyond bit $55$ in a $+2F$ correction, making the coefficient $2^{56}$ too small. A raw-word model matches all $702$ ordinary inverse boundaries. The default mixed circuit also fails, excluding disabled optimization passes as the cause; the conservative $w=4$ circuit passes.

\subsection{Targeted and Canonical Repairs}
\label{sec:canonical-reference}

\href{https://github.com/jieyilong/ecdsafail-circuit-evidence/tree/d6a93d1baaa2cc472278d2b419ef41d2d4708074/experiments/10-targeted-repair}{Experiment~10} retains approximate replay but adds zero-payload sign masking, equality-aware carry cleanup, canonical coordinate subtraction, and $768$ rounds. Both $64$-denominator zero-payload tests pass, as do the two zero-slope points and $64$ smoke inputs at $w=4$ and $16$. Its lower-cost operating point is in \Cref{app:fresh-artifacts}. Flag erasure requires the actual input and output words to preserve the zero predicate; the \href{https://github.com/jieyilong/ecdsafail-circuit-evidence/tree/d6a93d1baaa2cc472278d2b419ef41d2d4708074/experiments/14-repair-audit}{independent audit} does not certify the remaining approximate arithmetic.

The distinct \href{https://github.com/jieyilong/ecdsafail-circuit-evidence/tree/d6a93d1baaa2cc472278d2b419ef41d2d4708074/experiments/09-canonical-reference}{canonical reference} retains $736$ value rounds but uses payloads in $\{0,\ldots,p-1\}$. Negation acts only on nonzero words. Addition uses a $257$-bit sum, full correction, and reduction-flag cleanup from $[\mathrm{result}<\mathrm{source}]$. Doubling and halving are
\begin{equation}
 \mathsf{dbl}(t)=2t-p[2t\geq p],\qquad
 \mathsf{half}(t)=\frac{t+p(t\bmod2)}{2}.
\end{equation}
Their flags are recovered from output parity and $[\mathsf{half}(t)\geq(p+1)/2]$, respectively. Full-width carry corrections and inverse canonical addition implement the three coordinate subtractions.

\begin{proposition}[Conditional canonical replay]
If the value circuit coherently computes and uncomputes its transcript and terminal signed units with correct relative phases, canonical replay implements division and multiplication on every canonical payload, including zero, with clean replay workspace.
\end{proposition}
\begin{proof}
Each primitive preserves canonical encoding, reconstructs its flags from unchanged controls and final values, and applies exact carry-phase corrections. The field identities thus hold as word identities: division leaves equal sign-corrected words for XOR cleanup, and inverse replay clears the redundant coefficient. The hypothesis supplies value-register restoration.
\end{proof}

Replay replacement alone leaves a separate final-$y$ subtraction phase defect; canonical coordinate subtraction repairs it. The combined reference passes $576$ cell cases, both zero-payload component tests, and the zero-slope and smoke tests at both window widths. Its $w=16$ cost is $1{,}804$ qubits and $5{,}188{,}043$ static Toffolis. A separate frozen $4{,}096$-input pilot has no detected failures and mean $T\approx5.188$ million. Noncanonical words, finite value widths, the round budget, and inherited square arithmetic remain outside the proposition's guarantee.

\section{Follow-Up Arithmetic and Coherence Analysis}
\label{app:oral-followup}

\subsection{A State-Dependent Coherent Error Bound}
\label{app:coherent-error}

Let $\mathcal E(\rho)=\sum_mK_m\rho K_m^\dagger$ be a completely positive trace-preserving implementation and $U$ an ideal isometry on a common space, including exceptional encodings. The Kraus index includes all discarded registers. For a certified input projector $\Pi$, assume $K_m\Pi=\gamma_mU\Pi$, with input-independent $\gamma_m$ and $\sum_m|\gamma_m|^2=1$. Correct basis outputs and input-independent outcome probabilities are insufficient: Kraus operators $I/\sqrt2,Z/\sqrt2$ satisfy these weaker conditions but dephase $\ket{+}$.

For any input, including an external reference, put $q=\operatorname{Tr}[(I-\Pi)\rho]$, with reference identities implicit. Purify $\rho$ as $\ket\Psi=\ket g+\ket a$, where $\ket g=\Pi\ket\Psi$ and $\|a\|^2=q$. The isometries $V=\sum_mK_m\otimes\ket m$ and $V_0=U\otimes\sum_m\gamma_m\ket m$ agree on $g$, so
\begin{equation}
 \|(V-V_0)\ket\Psi\|=\|(V-V_0)\ket a\|\leq2\sqrt q.
\end{equation}
Trace distance is bounded by this vector norm and contracts when the environment is discarded. Replacing one call at a time, with ideal prefixes and implemented suffixes, gives $\min\{1,2\sum_i\sqrt{q_i}\}$ as in \Cref{eq:state-dependent-composition}. Each $q_i$ concerns the ideal pre-call state; no independence between calls is assumed. \href{https://github.com/jieyilong/ecdsafail-circuit-evidence/tree/d6a93d1baaa2cc472278d2b419ef41d2d4708074/experiments/11-coherent-error}{Experiment~11} provides the full proof and checks.

For uniformly prepared disjoint windows, $R_i=[c+\sum_{\ell<i}a_\ell j_\ell]G$, with $a_\ell=2^{s_\ell}$ or $k2^{s_\ell}$. For a computational-basis good set, $q_i$ counts bad labeled prefixes with multiplicity, not uniformly sampled curve points. After direct $16$-bit initialization in the low-to-high $G$ schedule, $R=[j_0]G$ and $A=[2^{16}j_1]G$ give unsupported $R=\mathcal O,A\neq\mathcal O$ weight $(2^{16}-1)/2^{32}$, not $1/r$. Different offsets or schedules change this value. The complete arithmetic bad set and its full-Shor weights remain uncertified.

\subsection{Schedule and Width Evidence}
\label{app:round-counterexamples}

The retained ten-million-denominator study has mean convergence $621.533239$, maximum $743$, and exceedance counts $1{,}535,6,0,0$ at budgets $704,736,768,800$. Original/conservative width misses are $661/0$ within their respective budgets. A zero count gives a pointwise one-sided $95\%$ upper bound near $3.00\times10^{-7}$ under independent sampling, not a universal bound. A post-hoc replay model finds no specified $40/72/48$-bit guard disagreements on $99{,}997$ nonidentity cases, but does not certify all phase cleanup or coordinate arithmetic.

The independent \href{https://github.com/jieyilong/ecdsafail-circuit-evidence/tree/d6a93d1baaa2cc472278d2b419ef41d2d4708074/experiments/12-schedule-audit}{schedule audit} retains the following exact witnesses, including initialization:
\begin{center}
\small
\setlength{\tabcolsep}{6pt}
\begin{tabular}{@{}lrr@{}}
\toprule
Denominator & Convergence transitions & First margin-$20$ width miss \\
\midrule
$1$ & $512$ & $177$ \\
$3$ & $1{,}135$ & $171$ \\
$2^{255}$ & $1{,}239$ & $80$ \\
\bottomrule
\end{tabular}
\end{center}
Width indices are zero-based. Since $|t'|\leq(|s|+|t|)/2$, ideal stored magnitudes stay at most $p$ and ordinary sums at most $2p$; signed $258$-bit arithmetic suffices for those sums, not the tapered circuit. A separate $10{,}000$-denominator sample has two walks exceeding $704$ and none exceeding $736/768$, with no margin-$20$ width misses through $768$. These studies are not pooled.

The witnesses lift to six supported affine pairs tested over eight measurement lanes each. The targeted $w=4$ repair fails all $48$ output lanes, with $25$ phase flags and clean final ancillas. Width violations precede round exhaustion for the long walks, so full-call failures do not isolate the mechanisms. The separate $0/100{,}000$ random result does not remove these counterexamples.

\subsection{Reversible Safegcd Reference}
\label{app:safegcd-reference}

\href{https://github.com/jieyilong/ecdsafail-circuit-evidence/tree/d6a93d1baaa2cc472278d2b419ef41d2d4708074/experiments/13-safegcd-reference}{Experiment~13} implements complete reversible division and multiplication. Safegcd uses the integer-$\delta$ recurrence and proved $741$-round bound at $n=256$~\cite{bernstein-yang19}, recording parity and conditional-swap decisions. Both backends use full signed widths, identical canonical payload primitives, and measurement-free Cuccaro arithmetic. Counts include initialization, replay, terminal correction, fixed input/output wires, and cleanup, but exclude the point-addition shell and lookup:
\begin{center}
\small
\setlength{\tabcolsep}{6pt}
\begin{tabular}{@{}lrrrr@{}}
\toprule
Reference & Rounds & $Q$ & Static Toffolis & Failures \\
\midrule
Safegcd & $741$ & $3{,}300$ & $7{,}501{,}932$ & $0/128$ \\
Ping-pong & $768$ & $2{,}320$ & $6{,}289{,}400$ & $12/128$ \\
Ping-pong & $1{,}536$ & $3{,}088$ & $12{,}574{,}712$ & $0/128$ \\
\bottomrule
\end{tabular}
\end{center}
Failures are per direction on the same denominator--payload cases, including zero payloads and long walks. All $12$ failures have incorrect outputs and nonzero scratch, but no phase flags. The $1{,}536$-round budget is not a proved all-input bound. Each construction also passes $5{,}196$ cases per direction over eight small primes with its stated test schedule; safegcd has additional coherent-state checks at $p=3,5,7,11$.

These unoptimized costs reverse their ordering when the ping-pong budget increases and establish no advantage over optimized safegcd. A separate $w=4$ shell test passes $64$ common inputs for both backends, but uses different payload lowering and therefore does not isolate the recurrence.

\end{document}